\documentclass[11pt]{article}

\usepackage{bm, commath}
\usepackage{natbib}
\usepackage{caption}
\usepackage{graphicx}
\usepackage{subfig}
\usepackage{amsmath, amsfonts, amsthm}
\usepackage{float}
\usepackage{booktabs,siunitx}
\usepackage{url}
\usepackage{multirow,multicol}
\usepackage{amssymb}
\usepackage{blkarray}
\usepackage{bbm}
\usepackage{multicol}
\usepackage{authblk}
\usepackage[shortlabels]{enumitem}

\usepackage{listings}
\usepackage{bbm}
\usepackage{textcomp}
\usepackage[margin=1in]{geometry}
\usepackage{authblk}
\usepackage[ruled]{algorithm2e}
\SetKwInput{KwParam}{Parameter}
\SetAlgoCaptionLayout{centerline}

\usepackage{sectsty}
\usepackage[doublespacing]{setspace}
\usepackage[utf8]{inputenc}
\usepackage{float}
\usepackage{tikz}
\usetikzlibrary{shapes,decorations,arrows,calc,arrows.meta,fit,positioning}

\usepackage[textsize=scriptsize, textwidth = 2cm, shadow]{todonotes}

\newtheorem{theorem}{Theorem}
\newtheorem{definition}{Definition}

\newtheorem{remark}{Remark}

\newtheorem*{assumption*}{Assumption}

\newtheorem{proposition}{Proposition}

\newtheorem{innercustomgeneric}{\customgenericname}
\providecommand{\customgenericname}{}
\newcommand{\newcustomtheorem}[2]{%
  \newenvironment{#1}[1]
  {%
   \renewcommand\customgenericname{#2}%
   \renewcommand\theinnercustomgeneric{##1}%
   \innercustomgeneric
  }
  {\endinnercustomgeneric}
}

\newcustomtheorem{customthm}{Theorem}
\newcustomtheorem{customassumption}{Assumption}

\usepackage{mathtools}

\usepackage{xcolor}

\makeatletter
\renewcommand{\algocf@captiontext}[2]{#1\algocf@typo. \AlCapFnt{}#2} % text of caption
\def\@algocf@capt@plain{top}
\renewcommand{\algocf@makecaption}[2]{%
  \addtolength{\hsize}{\algomargin}%
  \sbox\@tempboxa{\algocf@captiontext{#1}{#2}}%
  \ifdim\wd\@tempboxa >\hsize%     % if caption is longer than a line
  \hskip .5\algomargin%
  \parbox[t]{\hsize}{\algocf@captiontext{#1}{#2}}% then caption is not centered
  \else%
  \global\@minipagefalse%
  \hbox to\hsize{\box\@tempboxa}% else caption is centered
  \fi%
  \addtolength{\hsize}{-\algomargin}%
}
\makeatother

\begin{document}

\sectionfont{\bfseries\large\sffamily}%

\subsectionfont{\bfseries\sffamily\normalsize}%

%\begin{center}
%\noindent
%{\sffamily\bfseries\LARGE 
%}%

%\noindent

%\end{center}

\title{\bf Determining vaccine responders in the presence of baseline immunity using single-cell assays and paired control samples}

%A minimally adjusted $p$-value approach

\author[1]{Zhe Chen$^\dagger$}
\author[2]{Siyu Heng$^\dagger$}
%\author[3,4]{Erica Andersen-Nissen}
%\author[4]{Sharon Khuzwayo}
\author[3,4,5]{Asa Tapley}
\author[3]{Stephen De Rosa}
\author[3]{Bo Zhang$^{*}$}

\affil[1]{Department of Biostatistics, Epidemiology and Informatics, University of Pennsylvania, Philadelphia, USA}
\affil[2]{Department of Biostatistics, School of Global Public Health, New York University, New York, USA}
%\affil[3]{Cape Town HVTN Immunology Laboratory, Hutchinson Centre Research Institute of South Africa, Cape Town, South Africa}
\affil[3]{Vaccine and Infectious Disease Division, Fred Hutchinson Cancer Center, Seattle, USA}
\affil[4]{Division of Allergy and Infectious Diseases, Department of Medicine, University of Washington, Seattle, USA}
\affil[5]{Department of Medicine, University of Cape Town, Cape Town, South Africa}
\date{}

\maketitle

\let\thefootnote\relax\footnotetext{${\dagger}$Zhe Chen and Siyu Heng contributed equally to this work.}

\let\thefootnote\relax\footnotetext{${*}$Corresponding Author: Bo Zhang (bzhang3@fredhutch.org).}

\noindent
\textsf{{\bf Abstract}: A key objective in vaccine studies is to evaluate vaccine-induced immunogenicity and determine whether participants have mounted a response to the vaccine. Cellular immune responses are essential for assessing vaccine-induced immunogenicity, and single-cell assays, such as intracellular cytokine staining (ICS) and B-cell phenotyping (BCP), are commonly employed to profile individual immune cell phenotypes and the cytokines they produce after stimulation. In this article, we introduce a novel statistical framework for identifying vaccine responders using ICS data collected before and after vaccination. This framework incorporates paired control data to account for potential unintended variations between assay runs, such as batch effects, that could lead to misclassification of participants as vaccine responders or non-responders. To formally integrate paired control data for accounting for assay variation across different time points (i.e., before and after vaccination), our proposed framework calculates and reports two $p$-values, both adjusting for paired control data but in distinct ways: (i) the maximally adjusted $p$-value, which applies the most conservative adjustment to the unadjusted $p$-value, ensuring validity over all plausible batch effects consistent with the paired control samples' data, and (ii) the minimally adjusted $p$-value, which imposes only the minimal adjustment to the unadjusted $p$-value, such that the adjusted $p$-value cannot be falsified by the paired control samples' data. Minimally and maximally adjusted $p$-values offer a balanced approach to managing type-I error rates and statistical power in the presence of batch effects. We apply this framework to analyze ICS data collected at baseline and 4 weeks post-vaccination from the COVID-19 Prevention Network (CoVPN) 3008 study. Our analysis helps address two clinical questions: 1) which participants exhibited evidence of an incident Omicron infection between baseline and 4 weeks after receiving the final dose of the primary vaccination series, and 2) which participants showed vaccine-induced T cell responses against the Omicron BA.4/5 Spike protein.

}%

\vspace{0.3 cm}

\section{Introduction}
\label{sec: introduction}
A key task in vaccine research is to assess the immunogenicity of a vaccine regimen. Vaccination can activate T cells, which then produce cytokines that play a crucial role in mediating the immune response. The intracellular cytokine staining (ICS) assay is commonly used in vaccine studies to measure antigen-specific cytokine production in response to stimulation. This assay operates at the single-cell level, where each cell is classified as either positive or negative for a specific combination of cytokines (e.g., IFN-$\gamma$ and/or IL-2), depending on whether it expresses the cytokine profile of interest. The output of the ICS assay reports the number of positive cells that match a specific phenotype (e.g., antigen-specific CD4+ or CD8+ T cells expressing IFN-$\gamma$ and/or IL-2) relative to the total number of cells in the sample. The total cell count in a given sample can range from a thousand to a couple hundred thousand, while the number of cells of the subpopulation of interest (i.e., the positive cells) can vary from a few dozen to several hundred in a single assay run, depending on the specific scientific question being addressed.

Two common tasks in immunogenicity studies are: 1) determining whether a participant's sample shows any response to a peptide pool; and 2) determining whether a participant's sample exhibits a vaccine-induced response. To address the first task, the number of positive cells in a sample (out of the total cell count) at a pre-specified time point --- such as baseline or post-vaccination --- is compared to that from a paired negative control sample. In experimental science, a negative control refers to a sample that does not undergo the treatment or procedure being studied and is expected not to show any positive response. A participant is considered positive if their sample meets a predefined criterion. Common methods for defining positivity include ad hoc rules based on fold change, one-sided Fisher’s exact test applied to $2 \times 2$ contingency tables, and Bayesian mixed models approaches \citep{finak2014mixture}.

A typical scenario involves analyzing samples collected from a participant at both baseline ($T_0$) and a post-vaccination time point that is expected to represent peak immune response ($T_1$). The participant's response at each time point is assessed separately, and positivity is determined for both $T_0$ and $T_1$. Based on the results, a participant can be classified into one of four categories: positive at both baseline and post-vaccination ($++$), negative at baseline but positive post-vaccination ($-+$), positive at baseline but negative post-vaccination ($+-$), or negative at both time points ($--$).

\begin{figure}
\centering
\begin{minipage}{.66\textwidth}
  \centering
  \includegraphics[width=\linewidth]{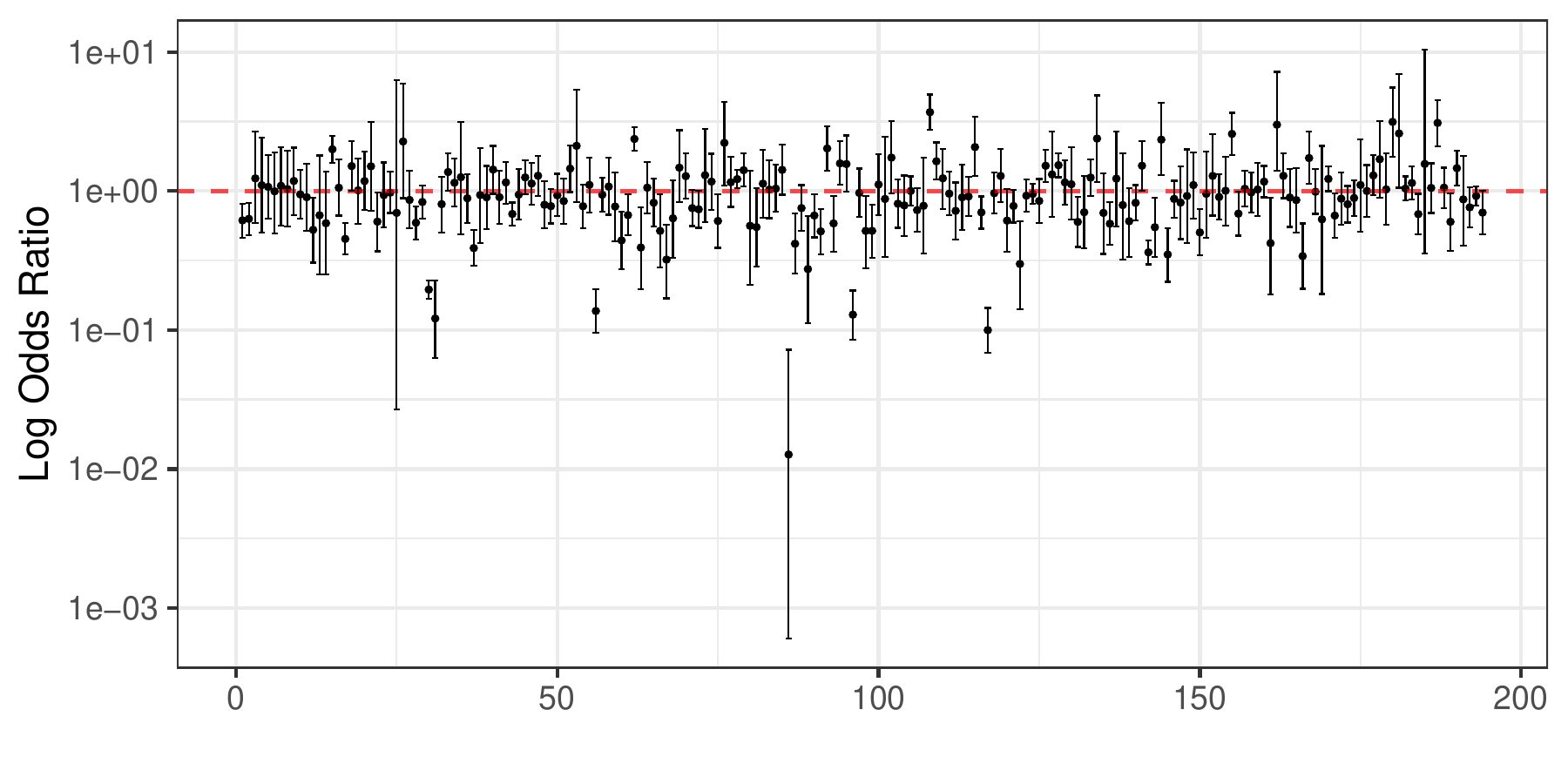}
\end{minipage}%
\begin{minipage}{.33\textwidth}
  \centering
  \includegraphics[width=\linewidth]{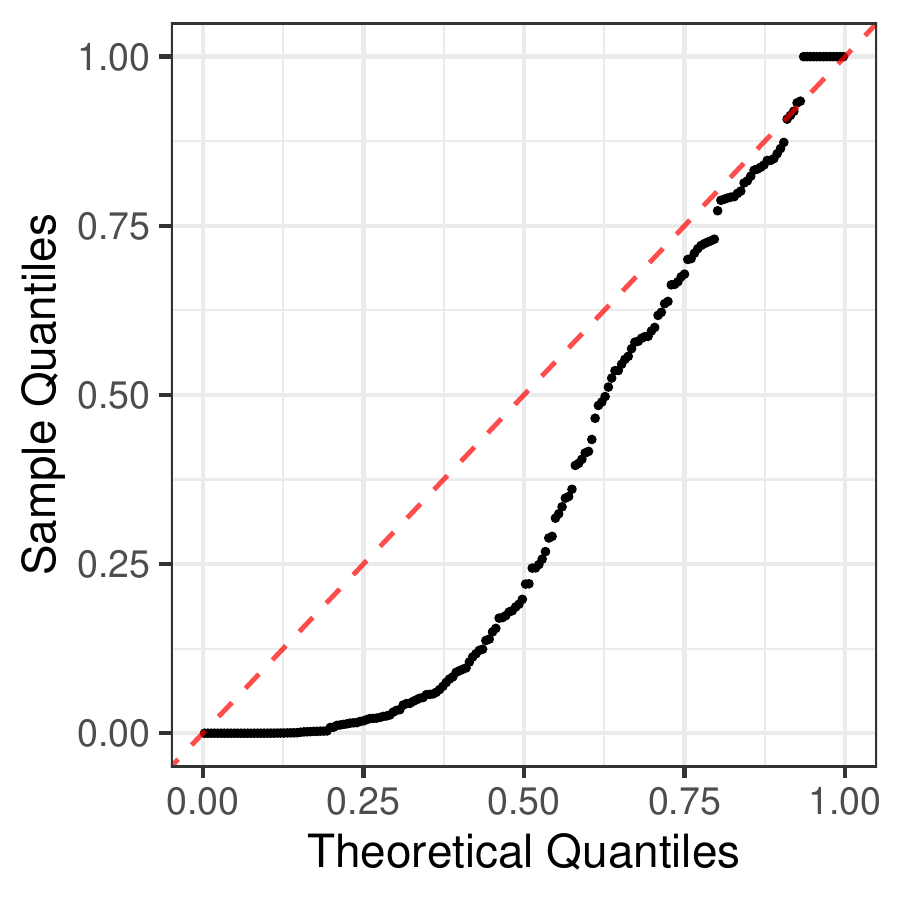}
\end{minipage}
\caption{Left panel: Estimated log odds ratio (baseline vs post-vaccination) using two-sided Fisher exact test from $N = 194$ negative control samples. Each vertical interval represents one participant's 95\% confidence interval for the log odds ratio. Among these, $65/194$ $(33.5\%)$ exhibit a significant difference between the two time points (before and after vaccination). Right panel: Sample quantiles of $N = 194$ $p$-values against a uniform distribution. The Q-Q plot suggests the P-values are anti-conservative.}
\label{fig: illustration}
\end{figure}

In addition to assessing whether a participant shows any immune response at baseline and post-vaccination, vaccine developers often aim to determine whether the vaccination effectively boosts the participant's immune response. In other words, they seek to identify vaccine responders. To answer this question, one might compare the post-vaccination level to the baseline level using methods such as Fisher’s exact test or tests for equal proportions. However, such comparisons can be biased due to batch effects when baseline and post-vaccination samples are measured in different assay runs \citep[Section 4.1]{finak2014mixture}. Specifically, the false positive rate (i.e., the probability that a negative cell is misclassified as positive) and the false negative rate (i.e., the probability that a positive cell is misclassified as negative) can vary between assay runs. For example, if the false positive rate is higher for post-vaccination samples than for baseline samples, participants could be falsely classified as vaccine responders despite having no genuine vaccine-induced response. In the context of outcome misclassification, this type of bias is referred to as differential outcome misclassification (\citealp{magder1997logistic}; \citealp{imai2010causal}; \citealp{chen2019inflation}; \citealp{edwards2023does}; \citealp{heng2025sensitivity}). 

Fortunately, in most experimental designs, researchers pair one or more control samples with a participant’s baseline and post-vaccination samples. By analyzing the number of positive cells in the control samples, one can better understand how batch effects may influence the comparison between post-vaccination and baseline levels. We illustrate this with data from $N = 194$ participants in the COVID-19 Vaccine and Prevention Trials Network (CoVPN) 3008 Study \citep{garretthybrid}, who provided peripheral blood mononuclear cell (PBMC) samples. For each participant, we analyzed two paired control samples: one paired with the baseline PBMC sample and the other with the post-vaccination PBMC sample. The left panel of Figure \ref{fig: illustration} shows the estimated confidence intervals for the log odds ratio from the two control samples of $N = 194$ participants, based on a two-sided Fisher’s exact test. In over $20\%$ of cases, the test would conclude that the control sample (unstimulated) at baseline has a lower expression of positive cells compared to the post-vaccination sample. In total, $65$ out of $194$ ($33.5\%$) tests yielded $p$-values significant at the $0.05$ level. The right panel of Figure \ref{fig: illustration} compares the distribution of these $p$-values to a uniform $\text{Unif}[0,1]$ distribution, revealing a deviation from the theoretical quantiles that would be expected in the absence of a batch effect.  

In this article, we examine a scenario where researchers have access to both pre-vaccination (baseline) and post-vaccination ICS data for antigenically non-naive participants. These individuals have previously been exposed to the antigen, either through natural infection or prior vaccination, and may, therefore, exhibit immunity at baseline. Our key question is whether the vaccination boosts the T cell response, as reflected by the ICS data measured at baseline and post-vaccination. We propose a novel framework for adjusting the $p$-value derived from a naive comparison between a participant's post-vaccination and baseline immune response levels by incorporating batch effect information extracted from paired control samples. Two quantities are central to our proposed framework: the maximally adjusted $p$-value and the minimally adjusted $p$-value. The maximally adjusted $p$-value represents the worst-case $p$-value, taking into account all plausible assay misclassification rates informed by the paired control samples. The minimally adjusted $p$-value, on the other hand, refers to the $p$-value that is closest to the unadjusted $p$-value among those that cannot be falsified by the batch effect information derived from the paired control samples. The minimally and maximally adjusted $p$-values differ in how they incorporate paired control samples, representing two endpoints on the spectrum of evidence for vaccine responders. Specifically, the maximally adjusted $p$-value applies the most conservative adjustment to the naive $p$-value, ensuring strict type-I error rate control. In contrast, the minimally adjusted $p$-value applies the least adjustment necessary to ensure that the adjusted $p$-value remains consistent with the paired control samples, while the unadjusted naive $p$-value may violate the batch effects indicated by the control data. The maximally adjusted $p$-value is most suitable when test validity is of paramount importance, such as in regulatory contexts, while the minimally adjusted $p$-value is better suited for exploratory objectives. In many application scenarios, we recommend reporting both the minimally and maximally adjusted $p$-values. This approach balances type-I error rate control and statistical power in the presence of batch effects, providing more comprehensive information for empirical researchers.

The rest of the article is organized as follows: Section \ref{sec: data and notation} outlines the general problem setup and introduces the notation used in our analysis. Section \ref{sec: method} describes the proposed methods. Section \ref{sec: simulation} reports simulation studies. Section \ref{sec: case study} applies the method to data from the CoVPN 3008 Study to identify Omicron infections and vaccine responders. We conclude with a brief discussion in Section \ref{sec: discussion}.

\section{Data and notation}
\label{sec: data and notation}
\subsection{Observed count data}
\label{subsec: data}
At baseline ($T_0$), an assay measures and reports $N_0$ total cell counts and $n_0 < N_0$ positive cell counts from participant $i$'s sample. Below, we will use ``positive" to denote classification of a single cell by the assay as expressing a biomarker of interest. \emph{In the same assay run}, the assay measures and reports $N'_0$ total cells and $n'_0 < N'_0$ positive cells on a paired control sample. Common types of control samples include: 1) negative control: a sample assayed under unstimulated conditions, where no positive response is expected; 2) positive control: a sample that is expected to produce a known positive response, confirming that the experimental system is functioning as intended; 3) generic control: a sample whose response is not expected to change between two timepoints (e.g., from $T_0$ to $T_1$), used to verify that no unintended variations occur over time; and 4) a combination of the above types. 

At a post-vaccination timepoint $T_1$, the assay measures and reports $n_1$ positive cells out of $N_1$ total cells for participant $i$. Similarly, on the paired control sample at $T_1$, the assay reports $n'_1$ positive cells out of $N'_1$ total cells. The typical data structure for these measurements is summarized in Table \ref{tab:2x2 table illustration}.

\begin{table}[ht]
    \centering
    \begin{tabular}{cccccc} \hline
     &\multicolumn{2}{c}{$\boldsymbol{T_0}$ (Pre-Vaccination)} &\multicolumn{2}{c}{$\boldsymbol{T_1}$ (Post-Vaccination)}\\
        \textbf{Sample Type}  & Positive & Negative & Positive & Negative\\
        Primary sample &  $n_0$ & $N_0 - n_0$ &  $n_1$ & $N_1 - n_1$ \\
       Control sample & $n'_0$ & $N'_0 - n'_0$ & $n'_1$ & $N'_1 - n'_1$ \\ \hline
    \end{tabular}
    \caption{Data structure of a typical single-cell assay output like the intracellular cytokine staining (ICS) assay for one participant.}
    \label{tab:2x2 table illustration}
\end{table}

We aim to determine whether participant $i$ is a vaccine responder. It is important to note that all parameters, including assay misclassification rates and the true proportion of positive cells, are unique to each participant and timepoint and may vary between participants. For clarity, the dependence of these quantities on $i$ will be understood implicitly and omitted in the subsequent discussion.

To make a vaccine responder call, one simple analysis approach is to conduct a one-sided test comparing the proportion of positive cells at baseline ($T_0$), based on count data $(n_0, N_0)$, with the proportion at a post-vaccination timepoint ($T_1$), based on count data $(n_1, N_1)$. In a typical ICS assay run, multiple functional parameters related to responding T cells are assessed simultaneously in the PBMC sample collected from a participant at a given timepoint. Samples from two different timepoints (e.g., baseline and post-vaccination) are analyzed in separate assay runs. A common concern in this approach is that even a well-validated assay may misclassify cells: negative cells may be incorrectly identified as positive, or vice versa. More critically, the misclassification rates are likely to differ between timepoints (i.e., differential misclassification of outcome), which can introduce batch effects and confound the analysis (\citealp{finak2014mixture}).

\subsection{Assay misclassification rates}
\label{subsec: data and notation misclassification}
We formalize the batch effect as follows. Let $p_{T_0, 1 \mid 0}$ denote the false positive rate at timepoint $T_0$, and $p_{T_0, 0 \mid 1}$ the false negative rate at timepoint $T_0$. Similarly, let $p_{T_1, 1 \mid 0}$ denote the false positive rate at timepoint $T_1$, and $p_{T_1, 0 \mid 1}$ the false negative rate at timepoint $T_1$. In general, we expect that $p_{T_0, 0 \mid 1} \neq p_{T_1, 0 \mid 1}$ and $p_{T_0, 1 \mid 0} \neq p_{T_1, 1 \mid 0}$, reflecting potential differences in misclassification rates between the two timepoints. We assume that the cells measured within the same timepoint and assay run share the same misclassification rates. Therefore, count data $(n_0, N_0)$ and $(n'_0, N'_0)$, measured at baseline ($T_0$), share the same misclassification rates $p_{T_0, 1 \mid 0}$ and $p_{T_0, 0 \mid 1}$, while count data $(n_1, N_1)$ and $(n'_1, N'_1)$, measured post-vaccination ($T_1$), share the same misclassification rates $p_{T_1, 1 \mid 0}$ and $p_{T_1, 0 \mid 1}$. Let $\boldsymbol{\theta} = (p_{T_0, 1 \mid 0}, p_{T_0, 0 \mid 1}, p_{T_1, 1 \mid 0}, p_{T_1, 0 \mid 1})$ represent the vector of misclassification rates.

Let $p^C_{T_0, 1}$ and $p^C_{T_0, 0} = 1 - p^C_{T_0, 1}$ represent the true proportion of positive and negative cells in the control sample at baseline ($T_0$), respectively. Similarly, let $p^C_{T_1, 1}$ and $p^C_{T_1, 0} = 1 - p^C_{T_1, 1}$ represent the true proportion of positive and negative cells in the control sample at post-vaccination ($T_1$). For a generic control, we assume that $p^C_{T_0, 1} = p^C_{T_1, 1} \geq 0$, meaning the true proportion of positive cells in the control sample is the same at both timepoints. For a negative control sample, we have $p^C_{T_0, 1} = p^C_{T_1, 1} = 0$, meaning no positive cells are expected at either timepoint. We define $p^{C, \ast}_{T_0, 1}$ and $p^{C, \ast}_{T_1, 1}$ as the expected proportions of observed positive cells (subject to misclassification) in the control sample at $T_0$ and $T_1$ under the misclassification rates $\boldsymbol{\theta}$, respectively. It is straightforward to verify the following relationships among $p^{C}_{T_0, 1}$, $p^{C, \ast}_{T_0, 1}$, $p^{C}_{T_1, 1}$, $p^{C, \ast}_{T_1, 1}$, and misclassification rates $\boldsymbol{\theta}$:
 \begin{equation*}
        p^{C}_{T_0, 1}=\frac{p^{C, \ast}_{T_0, 1} - p_{T_0, 1\mid 0}}{1-p_{T_0, 0\mid 1}-p_{T_0, 1\mid 0}}~~~~\text{and}~~~~ p^{C}_{T_1, 1}=\frac{p^{C, \ast}_{T_1, 1} - p_{T_1, 1\mid 0}}{1 - p_{T_1, 0\mid 1}-p_{T_1, 1\mid 0}}.
\end{equation*}

%In the case of a positive control, we further assume that $p^C_{T_0, 1} = p^C_{T_1, 1} > 0$, \textcolor{blue}{ZC: Do we assume $p^C_{T_0, 1} = p^C_{T_1, 1}$ for postive controls?} indicating that the true proportion of positive cells is nonzero at both timepoints.

Let $p_{T_0, 1}$ (or $p_{T_0, 0}$) and $p_{T_1, 1}$ (or $p_{T_1, 0}$) denote the true proportions of positive (or negative) cells in participant $i$'s primary sample at baseline ($T_0$) and post-vaccination ($T_1$), respectively. Let $p^{\ast}_{T_0, 1}$ (or $p^{\ast}_{T_0, 0}$) and $p^{\ast}_{T_1, 1}$ (or $p^{\ast}_{T_1, 0}$) represent the expected proportions of observed positive (or negative) cells (subject to misclassification) in the primary sample at $T_0$ and $T_1$, respectively, under the misclassification rates $\boldsymbol{\theta}$. Similarly, the following relationships among $p_{T_0, 1}$, $p^{\ast}_{T_0, 1}$, $p_{T_1, 1}$, $p^{\ast}_{T_1, 1}$, and misclassification rates $\boldsymbol{\theta}$ hold (\citealp{magder1997logistic}): 
\begin{equation*}
        p_{T_0, 1}=\frac{p^{\ast}_{T_0, 1} - p_{T_0, 1\mid 0}}{1-p_{T_0, 0\mid 1}-p_{T_0, 1\mid 0}}~~~~\text{and}~~~~ p_{T_1, 1}=\frac{p^{\ast}_{T_1, 1} - p_{T_1, 1\mid 0}}{1 - p_{T_1, 0\mid 1}-p_{T_1, 1\mid 0}}.
    \end{equation*}

\section{Incorporating the paired control data in determining a vaccine responder}\label{sec: method}
\subsection{Quantifying misclassification rates; vaccine responder $p$-value under misclassification}
\label{subsec: method quantify misclassification rates}

The first step in calculating either the maximally or minimally adjusted $p$-value is to use the observed paired control samples at $T_0$ and $T_1$ to construct an asymptotically valid confidence set for the misclassification rates. We first consider the setting with a generic control such that $p^C_{T_0, 1} = p^C_{T_1, 1} \geq 0$. We first examine two generic control samples at $T_0$ and $T_1$. Under fixed misclassification rates $\boldsymbol{\theta}$, we define the following test statistic: 
\begin{equation}\label{eq: T^C}
      T^{C}=\widehat{p}^{C}_{T_1, 1}-\widehat{p}^{C}_{T_0, 1} = \frac{\widehat{p}^{C, \ast}_{T_1, 1}-p_{T_1, 1\mid 0}}{1-p_{T_1, 0 \mid 1} - p_{T_1, 1 \mid 0}} - \frac{\widehat{p}^{C, \ast}_{T_0, 1}-p_{T_0, 1\mid 0}}{1- p_{T_0, 0 \mid 1} - p_{T_0, 1 \mid 0}},
\end{equation}
  where $\widehat{p}^{C, \ast}_{T_1, 1} = n'_1/N'_1$ and $\widehat{p}^{C, \ast}_{T_0, 1} = n'_0/N'_0$ denote the sample proportion of positive cells in the generic control samples at $T_1$ and $T_0$, respectively. Because we know $p^{C}_{T_0, 1} = p^{C}_{T_1, 1}$ \emph{a priori}, observed count data in the paired control samples can be used to construct a confidence set for the misclassification rates $\boldsymbol{\theta}$ at $T_0$ and $T_1$, as stated in Proposition~\ref{prop: Test for generic control}. %Specifically, consider testing the null hypothesis $H^C_0: \boldsymbol{\theta} = \boldsymbol{\theta_0}$ against the alternative $H^C_1: \boldsymbol{\theta} \neq \boldsymbol{\theta_0}$ based on the test statistic $T^C$ in Equation \eqref{eq: T^C}. The following result Proposition~\ref{prop: Test for generic control} gives an asymptotically valid .

\begin{proposition}\label{prop: Test for generic control}
Assume that the generic control samples are independent and identically distributed at both $T_{0}$ and $T_{1}$ and are independent across $T_{0}$ and $T_{1}$. Under invariance of true proportions of positive
cells for paired generic control sample at $T_0$ and $T_1$ (i.e., $p^{C}_{T_0, 1} = p^{C}_{T_1, 1}$) and each fixed vector of misclassification rates $\boldsymbol{\theta} = (p_{T_0, 1 \mid 0}, p_{T_0, 0 \mid 1}, p_{T_1, 1 \mid 0}, p_{T_1, 0 \mid 1})$, we have
\begin{equation*}
      \overline{T}^{C}(\boldsymbol{\theta}) =  \frac{T^C}{\sqrt{\widehat{p}^{C}_{1}(1-\widehat{p}^{C}_{1}) \Big(\frac{1}{N_{1}^{\prime}}+\frac{1}{N_{0}^{\prime}}\Big) } }\xrightarrow{d} N(0,1) \text{ as $\min\{N_{0}^{\prime}, N_{1}^{\prime}\}\rightarrow \infty$},
\end{equation*}
where $T^C$ is given in \eqref{eq: T^C}, and
\begin{equation*}
\widehat{p}^{C}_{1}=\frac{N_{1}^{\prime}}{N_{0}^{\prime}+N_{1}^{\prime}}\times \frac{\widehat{p}^{C, \ast}_{T_1, 1}-p_{T_1, 1\mid 0}}{1-p_{T_1, 0 \mid 1}-p_{T_1, 1 \mid 0}} + \frac{N_{0}^{\prime}}{N_{0}^{\prime}+N_{1}^{\prime}}\times \frac{\widehat{p}^{C, \ast}_{T_0, 1}-p_{T_0, 1\mid 0}}{1-p_{T_0, 0 \mid 1}-p_{T_0, 1 \mid 0}}.
\end{equation*}
Then, consider the confidence set $\mathcal{A}(\alpha)=\{ \boldsymbol{\theta}_{0}\in [0,1]^{4}:  |\overline{T}^{C}(\boldsymbol{\theta}_{0})|\leq  \Phi^{-1}(1-\alpha/2)\}$ for misclassification rates $\boldsymbol{\theta}$, in which $\Phi$ is the distribution function of $N(0,1)$. As $\min\{N_{0}^{\prime}, N_{1}^{\prime}\}\rightarrow \infty$, the coverage rate of $\mathcal{A}(\alpha)$ is asymptotically no less than $100\times(1-\alpha)\%$. 

\end{proposition}

\begin{proof}
    All proofs in the article can be found in the online supplemental materials.
\end{proof}

\begin{remark}[Paired negative control samples]
   If the paired controls are negative controls, such that $p^C_{T_0, 1} = p^C_{T_1, 1} = 0$, then the observed negative control data contain no information about the false negative rate at either $T_0$ or $T_1$, because negative cells in the paired negative controls are always true negatives. In this case, two assay false negative rates should be regarded as sensitivity parameters. The confidence set $\mathcal{A}(\alpha)$ for $\boldsymbol{\theta}$ then reduces to a confidence set for two assay false positive rates, which can be constructed as the union of 1) the $100\times(1-\alpha/2)\%$ confidence set for $p_{T_0, 1 \mid 0}$ at $T_0$, based on $(n'_0, N'_0)$; and 2) the $100\times(1-\alpha/2)\%$ confidence set for $p_{T_1, 1 \mid 0}$ at $T_1$, based on $(n'_1, N'_1)$.
\end{remark}

Next, we focus on the primary samples at $T_0$ and $T_1$. A study participant is considered a \emph{vaccine responder} at level-$\alpha$ if the null hypothesis
\begin{equation}
    \label{eq: primary null}
    H_0: p_{T_1, 1} = p_{T_0, 1}
\end{equation}
is rejected at level-$\alpha$ against the one-sided alternative $H_1: p_{T_1, 1} > p_{T_0, 1}$. Specifically, under fixed misclassification rates $\boldsymbol{\theta}$, we can test the null hypothesis $H_0$ using the following test statistic $T$:
  \begin{equation}\label{eq: T}
      T = \widehat{p}_{T_1, 1}-\widehat{p}_{T_0, 1} = \frac{\widehat{p}^{\ast}_{T_1, 1} - p_{T_1, 1\mid 0}}{1-p_{T_1, 0\mid 1}-p_{T_1, 1\mid 0}} -\frac{\widehat{p}^{\ast}_{T_0, 1} - p_{T_0, 1\mid 0}}{1-p_{T_0, 0\mid 1}-p_{T_0, 1\mid 0}},
  \end{equation}
  where $\widehat{p}^{\ast}_{T_1, 1} = n_1/N_1$ and $\widehat{p}^{\ast}_{T_0, 1} = n_0/N_0$ denote the sample proportion of positive cells in participant $i$'s primary sample at $T_1$ and that at $T_0$, respectively.

\begin{proposition}
\label{prop: study sample}
Assume that the primary assay samples (i.e., study samples) are independent and identically distributed at both $T_{0}$ and $T_{1}$ and are independent across $T_{0}$ and $T_{1}$. Under $H_{0}$ and each fixed misclassification rates $\boldsymbol{\theta}$, we have 
    \begin{equation*}
      \overline{T}(\boldsymbol{\theta})=  \frac{T}{\sqrt{\widehat{p}_{1}(1-\widehat{p}_{1}) \Big(\frac{1}{N_{1}}+\frac{1}{N_{0}}\Big) } }\xrightarrow{d} N(0,1) \text{ as $\min\{N_{0}, N_{1}\}\rightarrow \infty$},
    \end{equation*}
    where $T$ is given in \eqref{eq: T}, and
    \begin{equation*}
        \widehat{p}_{1}=\frac{N_{1}}{N_{0}+N_{1}}\times \frac{\widehat{p}^{\ast}_{T_1, 1}-p_{T_1, 1\mid 0}}{1-p_{T_1, 0\mid 1}-p_{T_1, 1\mid 0} }+\frac{N_{0}}{N_{0}+N_{1}}\times \frac{\widehat{p}^{\ast}_{T_0, 1}-p_{T_0, 1\mid 0}}{1-p_{T_0, 0\mid 1}-p_{T_0, 1\mid 0}}.
    \end{equation*}
\end{proposition}
\noindent  Under Proposition~\ref{prop: study sample}, an asymptotically valid one-sided $p$-value for testing $H_{0}$ against $H_1$ given true misclassification rates $\boldsymbol{\theta}$ is $p_{\boldsymbol{\theta}}= 1-\Phi(\overline{T}(\boldsymbol{\theta}))$. In the absence of cell misclassification (i.e., when \(\boldsymbol{\theta} = \boldsymbol{0}\)), Proposition \ref{prop: study sample} reduces to a standard two-sample proportion testing procedure that reports an unadjusted $p$-value, denoted as $p^{*}$, ignoring the potential for cell misclassification. 

\subsection{Combining Propositions~\ref{prop: Test for generic control} and \ref{prop: study sample}: three motivating examples}
\label{subsec: three examples}

How can we leverage the confidence set $\mathcal{A}(\alpha)$ of misclassification rates $\boldsymbol{\theta}$ in Proposition~\ref{prop: Test for generic control}, derived from the paired control data, to inform decision-making in the vaccine responder analysis presented in Proposition~\ref{prop: study sample}? First, for values of $\boldsymbol{\theta}$ that are not in $\mathcal{A}(\alpha)$, we consider them incompatible with the observed data. Therefore, $p$-values derived from comparing primary responses at $T_1$ vs. $T_0$ under $\boldsymbol{\theta} \notin \mathcal{A}(\alpha)$ should not be considered or used to inform decision-making. The remaining question is how to incorporate the confidence set $\mathcal{A}(\alpha)$ --- and the vaccine responder $p$-values corresponding to misclassification rates in $\mathcal{A}(\alpha)$ --- to adjust the naive $p$-value that ignores the potential assay misclassification rates. 

To motivate our proposed framework to answer this question, we present three concrete motivating examples. Table \ref{tab: case study example special case} summarizes the observed case counts in the primary and paired control samples for three participants. For ease of illustration, we temporarily assume that the assay false negative rates at both $T_0$ and $T_1$ are zero, that is, $(p_{T_0, 0 \mid 1}, p_{T_1, 0 \mid 1}) = (0, 0)$. This assumption allows us to focus on the false positive rates, $(p_{T_0, 1 \mid 0}, p_{T_1, 1 \mid 0})$, and simplifies the visualization of the analysis.

\begin{table}[ht]
    \centering
    \begin{tabular}{ccccccc} \hline
     &\multicolumn{2}{c}{$T_0$ (Baseline)} & &\multicolumn{2}{c}{$T_1$ (Post-vaccination)}\\ \cline{2-3} \cline{5-6}
      \textbf{Participant 1}  & Positive & Negative && Positive & Negative\\
        Participant $1$'s sample &  $31$ & $69{,}540 - 31$  &&  $85$ & $93{,}562 - 85$ \\
       Paired control sample& $8$ & $93{,}883 - 8$ && $43$ & $212{,}650 - 43$ \\ \hline
       \textbf{Participant 2}  & Positive & Negative && Positive & Negative\\
        Participant $2$'s sample &  $31$ & $69{,}540 - 31$ &&  $85$ & $93{,}562 - 85$ \\
       Paired control sample & $8$ & $93{,}883 - 8$ && $15$ & $212{,}650 - 15$ \\ \hline
       \textbf{Participant 3}  & Positive & Negative && Positive & Negative\\
        Participant $3$'s sample &  $31$ & $69{,}540 - 31$ &&  $85$ & $93{,}562 - 85$ \\
       Paired control sample& $8$ & $93{,}883 - 8$ && $2$ & $212{,}650 - 2$ \\ \hline
    \end{tabular}
    \caption{Three examples of patient-level ICS assay readouts for the primary sample and the paired control sample at $T_0$ and $T_1$.}
    \label{tab: case study example special case}
\end{table}

For Participant 1, the paired control samples suggest evidence supporting a larger false positive rate at $T_1$ compared to that at $T_0$. The shaded area in the left panel of Figure \ref{fig: case study special} corresponds to the 95\% confidence set $\mathcal{A}(0.05)$ obtained according to Proposition \ref{prop: Test for generic control}, and the vaccine responder $p$-values corresponding to misclassification rates in $\mathcal{A}(0.05)$ range from $4\times 10^{-4}$ to $8.9\times 10^{-3}$ (we denote this range of $p$-values as $\mathcal{P}_{\mathcal{A}(0.05)}$). For Participant $1$, $(p_{T_0, 1\mid 0}, p_{T_1, 1\mid 0}) = (0, 0)$ is not contained in $\mathcal{A}(0.05)$, and an unadjusted analysis assuming $(p_{T_0, 1\mid 0}, p_{T_1, 1\mid 0}) = (0, 0)$ would yield a $p$-value of $3\times 10^{-4}$, which is smaller than the smallest $p$-value in $\mathcal{P}_{\mathcal{A}(0.05)}$. For this participant, a vaccine-responder call based on the unadjusted $p$-value does not acknowledge the ample evidence supporting a batch effect. 

For Participant 2, the paired control samples do not provide evidence of differential assay false positive rates at $T_1$ compared to $T_0$, as shown in the middle panel of Figure \ref{fig: case study special}. The confidence set $\mathcal{A}(0.05)$ includes the combination $(p_{T_0, 1 \mid 0}, p_{T_1, 1 \mid 0}) = (0, 0)$. The vaccine responder $p$-values range from $2 \times 10^{-5}$ to $6 \times 10^{-4}$, which includes the unadjusted $p$-value of $3 \times 10^{-4}$. 

For Participant 3, the paired control samples suggest evidence of a larger false positive rate at $T_0$ compared to $T_1$, as depicted in the right panel of Figure \ref{fig: case study special}. In this case, the confidence set $\mathcal{A}(0.05)$ does not include the combination $(p_{T_0, 1 \mid 0}, p_{T_1, 1 \mid 0}) = (0, 0)$. The vaccine responder $p$-values corresponding to misclassification rates in $\mathcal{A}(0.05)$ range from $1 \times 10^{-5}$ to $6 \times 10^{-5}$, all of which are more significant than the unadjusted $p$-value of $3 \times 10^{-4}$. 

%For Participant 3, we recommend reporting the adjusted $p$-value of $6 \times 10^{-5}$, which is the most conservative $p$-value among those compatible with the observed control data.

\begin{figure}[ht]
    \centering
   \begin{minipage}{.33\textwidth}
        \centering
\includegraphics[width=0.98\linewidth]{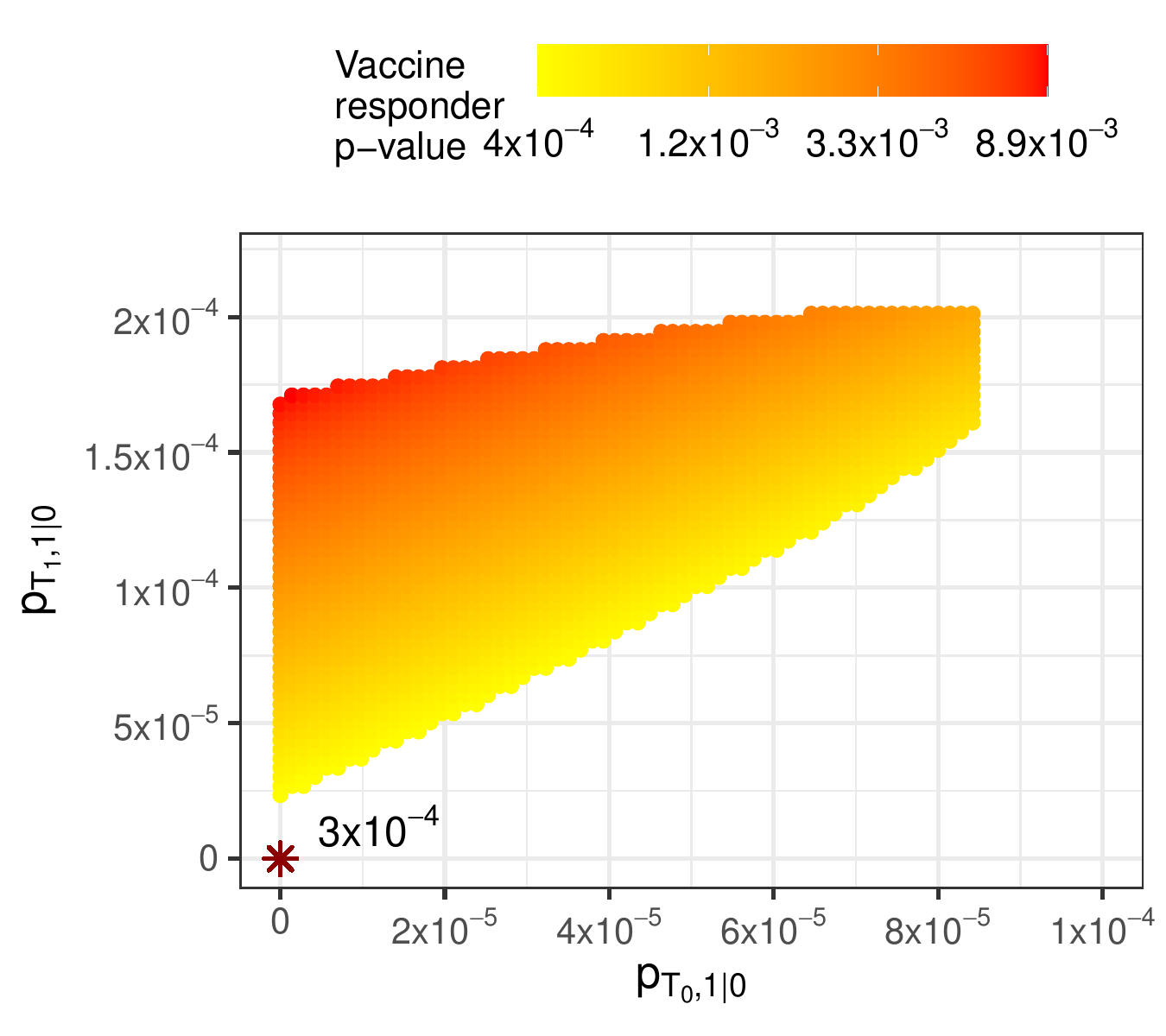}
    \end{minipage}%
    \begin{minipage}{0.33\textwidth}
        \centering
\includegraphics[width=0.98\linewidth]{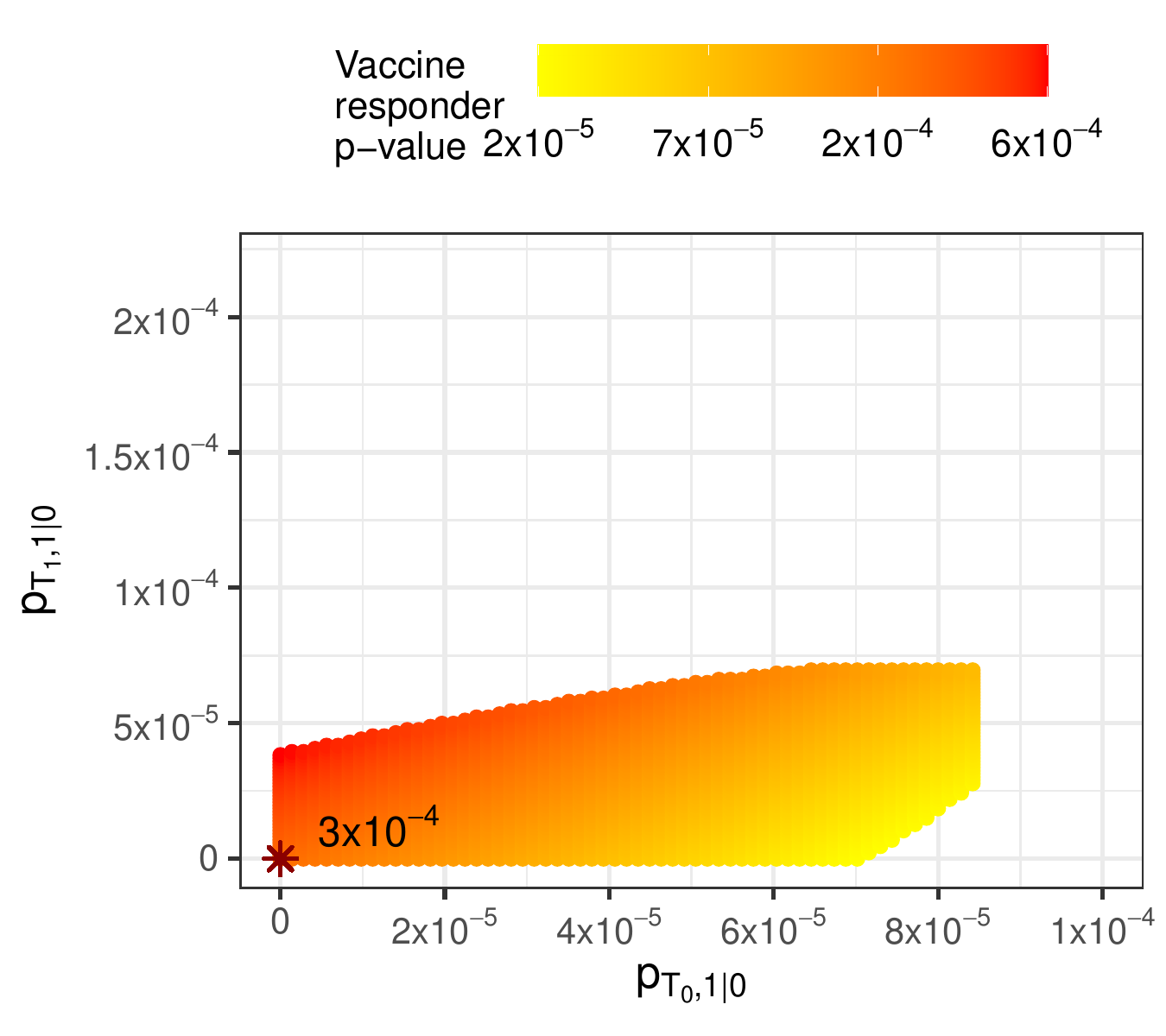}
    \end{minipage}
     \begin{minipage}{0.33\textwidth}
        \centering
\includegraphics[width=0.98\linewidth]{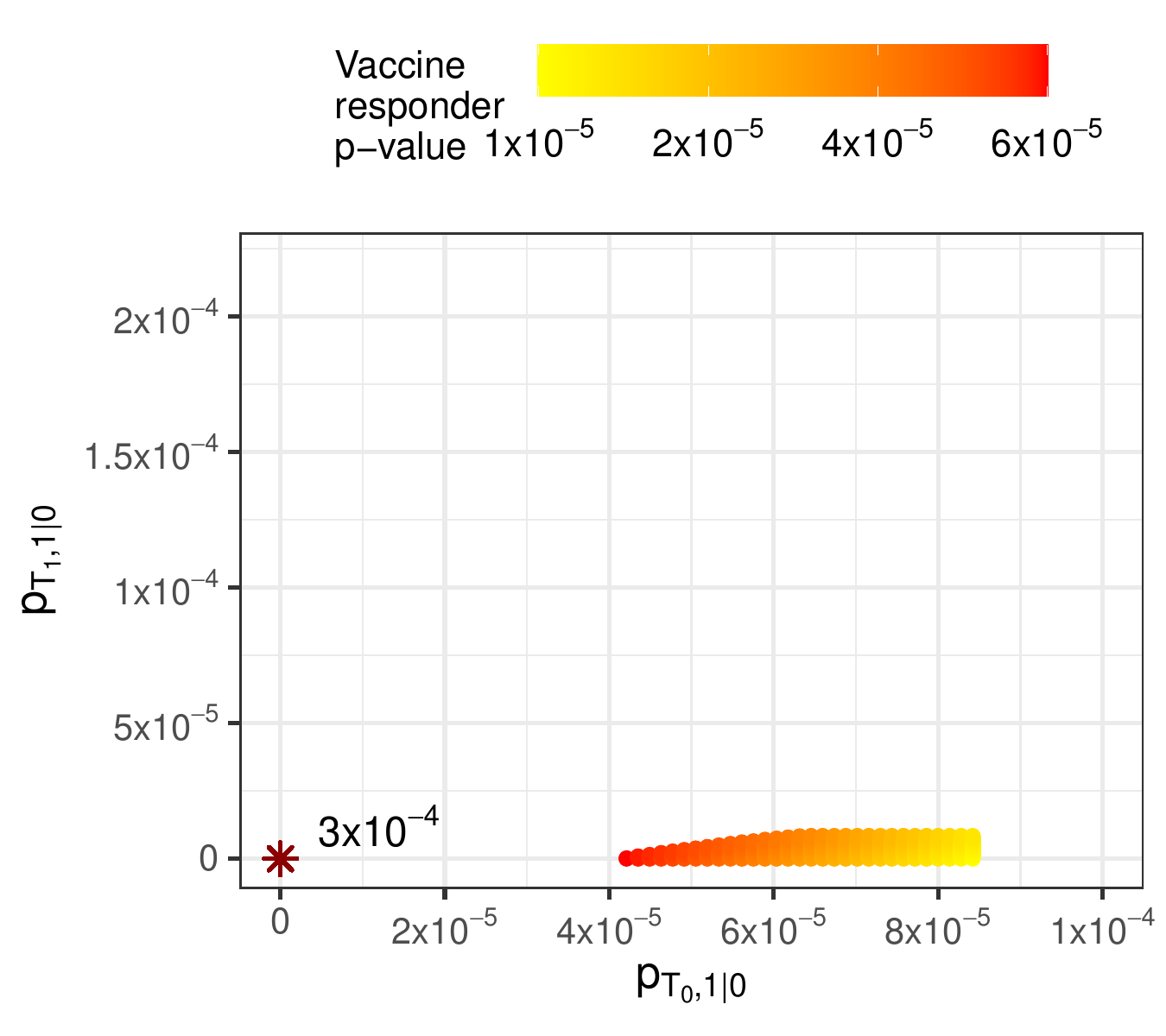}
    \end{minipage}
    \caption{Confidence set of $(p_{T_0, 1\mid 0}, p_{T_1, 1\mid 0})$ and vaccine responder $p$-values corresponding to different combinations of $(p_{T_0, 1\mid 0}, p_{T_1, 1\mid 0})$ in Example 1 (left panel), Example 2 (middle panel), and Example 3 (right panel). The unadjusted vaccine responder $p$-value, corresponding to $(p_{T_0, 1\mid 0}, p_{T_1, 1\mid 0}) = (0, 0)$, equals $3\times 10^{-4}$. }
    \label{fig: case study special}
\end{figure}

\subsection{The proposed framework: maximally and minimally adjusted $p$-values}
\label{subsec: method minimally adjusted p-value}

One strategy to account for the confidence set $\mathcal{A}(\alpha)$ for the misclassification rates $\boldsymbol{\theta}$ (constructed in Proposition~\ref{prop: Test for generic control}) is to maximize the $p$-value for testing $H_{0}$ (see Proposition~\ref{prop: study sample}) over all $\boldsymbol{\theta}\in \mathcal{A}(\alpha)$. We refer to this maximal $p$-value as the \emph{maximally adjusted $p$-value}, and its formal definition is provided in Definition~\ref{def: BB} below.

%One approach to accounting for the nuisance parameters in statistical inference is the \citeauthor{berger1994p} \citeyearpar{berger1994p} procedure of reporting the maximum $p$-value over the confidence set of the nuisance parameters. This motivates our first definition of a vaccine responder $p$-value, termed maximally adjusted $p$-value, which accounts fully for the assay misclassification rates information contained in the paired negative control samples.

\begin{definition}[Maximally adjusted $p$-value]\label{def: BB}
    Let $\mathcal{A}(\alpha')$ denote a $100\times(1-\alpha')\%$ confidence set of the assay misclassification rates $\boldsymbol{\theta}$ based on the paired control data, as constructed in Proposition~\ref{prop: Test for generic control}, and $p_{\boldsymbol{\theta}}$ the vaccine responder $p$-value under a fixed $\boldsymbol{\theta}$, as derived in Proposition~\ref{prop: study sample}. Define $\mathcal{P}_{\mathcal{A}(\alpha')} = \{p_{\boldsymbol{\theta}}: \boldsymbol{\theta} \in \mathcal{A}(\alpha')\}$. Then, the maximally adjusted $p$-value is
    \begin{equation*}
        p_{\text{max}}^{\text{adj}}=\sup \mathcal{P}_{\mathcal{A}(\alpha')}+\alpha^{\prime}.
    \end{equation*}
\end{definition}
\noindent Similar to \citeauthor{berger1994p}' argument (\citealp{berger1994p}), the maximally adjusted $p$-value $p_{\text{max}}^{\text{adj}}$ can be shown to be asymptotically valid without knowing the true misclassification rates $\boldsymbol{\theta}$.

\begin{theorem}\label{thm: validity of max adj p}(Validity of the maximally adjusted $p$-value) Under the same setting as those in Propositions~\ref{prop: Test for generic control} and \ref{prop: study sample} and the assumption that the set $A(\alpha^{\prime})$ is compact, for any prespecified significance level $\alpha\in (0,1)$, and for any prespecified $\alpha^{\prime}<\alpha$, as $\min\{N_{0}^{\prime}, N_{1}^{\prime}\}\rightarrow \infty$, we have $\lim \sup P(p_{\text{max}}^{\text{adj}}\leq \alpha \mid H_{0})\leq \alpha$. 
\end{theorem}

%The advantage of the maximally adjusted $p$-value $p_{\text{max}}^{\text{adj}}$, as shown in Theorem~\ref{thm: validity of max adj p}, is its strict validity in the absence of the knowledge about the true assay misclassification rates $\boldsymbol{\theta}$. 
Reporting $p_{\text{max}}^{\text{adj}}$ is most suitable in two scenarios. First, it is suitable when the paired control data provide informative insights into the misclassification rates (i.e., when the confidence set $A(\alpha^{\prime})$ is informative). Second, it is appropriate when a valid $p$-value is critical for decision-making, even at the cost of statistical power.

On the other hand, in some experimental settings, paired control data --- especially those from generic controls --- are primarily used to monitor changes in experimental conditions between two assay runs (e.g., via using a Levey-Jennings control chart; see \citet{levey1950use}) rather than to accurately estimate assay misclassification rates. In such cases, solely reporting the maximally adjusted $p$-value may make our analysis unnecessarily conservative and less informative, particularly for exploratory purposes. To address this, we introduce another type of adjusted $p$-value called the \emph{minimally adjusted $p$-value}. This approach applies only a minimal adjustment to the unadjusted $p$-value, yet the adjustment is sufficient to ensure that the adjusted $p$-value does not violate the paired control data. Importantly, it performs no worse than the unadjusted $p$-value, whereas a naive, unadjusted $p$-value could potentially be falsified by the paired control data.

To further motivate the idea of a minimally adjusted $p$-value, recall the examples in Section \ref{subsec: three examples}. For Participant $1$, among all $(p_{T_0, 1 \mid 0}, p_{T_1, 1 \mid 0})$ combinations compatible with the paired control data, the smallest $p$-value is $4 \times 10^{-4}$, which corresponds to an assay misclassification rates configuration that is not invalidated by the observed control data with high probability, while maintaining a similar significance level of vaccine responder evidence as the unadjusted $p$-value. As another example, for Participant $2$, given the lack of evidence supporting a batch effect ($\boldsymbol{\theta} = \boldsymbol{0}$ cannot be rejected at any reasonable level), reporting the unadjusted $p$-value of $3 \times 10^{-4}$ does not violate paired control data in this case. Motivated by these observations, the concept of the minimally adjusted $p$-value is to provide a $p$-value that cannot be falsified by the paired control data, while remaining as consistent as possible with the evidence indicated by the unadjusted $p$-value, \( p^\ast \).

\begin{definition}[Minimally adjusted $p$-value]\label{def: minimally adjusted} Consider the $\mathcal{A}(\alpha)$, $p_{\boldsymbol{\theta}}$, and $\mathcal{P}_{\mathcal{A}(\alpha)} = \{p_{\boldsymbol{\theta}}: \boldsymbol{\theta} \in \mathcal{A}(\alpha)\}$ as defined in Definition \ref{def: BB}. Assume $\mathcal{P}_{\mathcal{A}(\alpha)} \neq \emptyset$ and is a closed set. Let $p^\ast$ denote the unadjusted $p$-value for testing $H_{0}$ assuming $\boldsymbol{\theta} = \boldsymbol{0}$ (i.e., ignoring the assay misclassification rates). A minimally adjusted $p$-value $p^{\text{adj}}_{\text{min}}$ is defined as, among all the $p$-values in $\mathcal{P}_{\mathcal{A}(\alpha)}$ (i.e., all the $p$-values consistent with the paired control data), the one with the minimal distance to the unadjusted $p$-value $p^\ast $, i.e., 
   \[
    p^{\text{adj}}_{\text{min}} = \text{argmin}_{p_{\boldsymbol{\theta}} \in \mathcal{P}_{\mathcal{A}(\alpha)}} |p^\ast - p_{\boldsymbol{\theta}}|.
\]
Note that when $p^\ast \in \mathcal{P}_{\mathcal{A}(\alpha)}$, we have $ p^{\text{adj}}_{\text{min}} = p^\ast$. 
\end{definition}

We give three remarks. First, when the confidence set $\mathcal{A}(\alpha)$ does not contain $\boldsymbol{\theta} = \boldsymbol{0}$, the set of $p$-values, $\mathcal{P}_{\mathcal{A}(\alpha)}$, may still contain the unadjusted $p$-value, $p^\ast$. In this case, the minimally adjusted $p$-value still equals the unadjusted $p$-value, i.e., $p^{\text{adj}}_{\text{min}} = p^\ast$. Second, if $\mathcal{P}_{\mathcal{A}(\alpha)} = \emptyset$, the minimally adjusted $p$-value is undefined. This situation suggests that the assumption of shared assay misclassification rates between the primary sample and the paired control sample may have been violated. Third, unlike the maximally adjusted $p$-value, the minimally adjusted $p$-value does not enforce strict type-I error rate control. However, as shown in Theorem~\ref{thm: min adj better than unadj}, with high probability, $p^{\text{adj}}_{\text{min}}$ is closer to the true $p$-value (i.e., the $p$-value for testing $H_{0}$ under the oracle misclassification rates) compared to the unadjusted $p$-value $p^{*}$.

\begin{theorem}\label{thm: min adj better than unadj}
    If the set $\mathcal{P}_{\mathcal{A}(\alpha)}$ is pathwise-connected, then the minimally adjusted $p$-value, $ p^{\text{adj}}_{\text{min}}$, with probability at least $100\times(1-\alpha)\%$, is asymptotically no worse than the unadjusted $p$-value, $p^\ast$, in the sense that $\lim P(| p^{\text{adj}}_{\text{min}} - p^{\text{truth}}| \leq |p^\ast - p^{\text{truth}}|\mid H_{0})\geq 1-\alpha$ as $\min\{N_{0}^{\prime}, N_{1}^{\prime}\}\rightarrow \infty$, where $p^{\text{truth}}$ is the vaccine responder $p$-value under $H_{0}$ and the true values of misclassification rates $\boldsymbol{\theta}$.
\end{theorem}

%\subsection{Practical issues}
%\label{subsec: practical guidance}
%\textcolor{purple}{This section needs some re-writing.}
%The adjusted $p$-value involves two hypothesis tests: one for the misclassification rates, based on the control samples, at level-$\alpha'$ and another for the vaccine responder $p$-value, based on the primary samples, at level-$\alpha''$. In most applications, the sample sizes in both the primary and control groups are large, so the analysis is typically not very sensitive to how the level-$0.05$ significance is split between $\alpha'$ and $\alpha''$, as long as $\alpha' + \alpha'' = \alpha$. We will illustrate this in the case study.

%When the paired control is a negative control, then the paired control data contain no information about the assay false negative rates. In this case, the assay false negative rates could be regarded as sensitivity parameters.

For a generic control, without imposing any structure on the four assay misclassification rates parameters, the confidence set $\mathcal{A}(\alpha)$ could be uninformative in the sense that 1) a participant will almost always be declared a non-responder according to Definition \ref{def: BB}; and 2) the minimally adjusted $p$-value will almost always equal the unadjusted $p$-value. When implementing either the maximally or minimally adjusted $p$-value as introduced in Definition \ref{def: BB} or \ref{def: minimally adjusted}, we recommend adding the following constraint on the difference in the assay false negative rates at $T_0$ and $T_1$:  $|p_{T_0, 0 \mid 1} - p_{T_1, 0 \mid 1}| \leq \delta_0$ for some pre-specified, reasonable $\delta_0$. For instance, one may stipulate that the assay false negative rates be the same at $T_0$ and $T_1$, corresponding to letting $\delta_0 = 0$, and reduce the degree of freedom of nuisance parameters $\boldsymbol{\theta}$ from $4$ to $3$.

\section{Simulation studies}
\label{sec: simulation}

\subsection{Simulation setup}
\label{subsec: simulation setup}

We evaluate the performance of the proposed methods across a range of data-generating processes. The simulation is designed to mimic realistic scenarios where baseline and post-vaccination measurements of cell positivity are subject to assay misclassification. The following steps outline the data-generating process of the primary and paired control samples at baseline ($T_0$) and post-vaccination ($T_1$) timepoints:

\paragraph{Responder status:}
Each participant is randomly assigned a binary true responder status, $R \in \{0, 1\}$, where $R = 1$ indicates a vaccine responder, with $P(R = 1) = 0.5$.

\paragraph{Primary sample size:}
The primary sample size is fixed at  $N_0 = N_1 =  50{,}000 $ for both baseline $T_{0}$ and post-vaccination $T_{1}$. 

\paragraph{Control sample size:}
 We consider the following sample size for the control sample: $N'_0 = N'_1 = 1000$, $N'_0 = N'_1 = 10{,}000 $, $N'_0 = N'_1 =  50{,}000 $, and $N'_0 = N'_1 = 100{,}000$.

\paragraph{Proportion of positive cells in the primary sample:}
The proportion of positive cells in the primary sample at baseline and post-vaccination, denoted by $p_{T_0, 1}$ and $p_{T_1, 1}$, respectively, are generated as follows:
\begin{align*}
    p_{T_0, 1} \sim \text{Beta}(1, 500),~~
    p_{T_1, 1} = 
    \begin{cases}
        p_{T_0, 1}, & \text{if } R = 0, \\
        \min\{\gamma \cdot p_{T_0, 1}, ~1\}, & \text{if } R = 1,
    \end{cases}
\end{align*}
where  $\gamma \in \{2,4,10\}$ controls the effect size for responders. 

\paragraph{Proportion of positive cells in the control sample:}
The proportions of positive cells in the control samples are constant across two timepoints, with $p^C_{T_0, 1} = p^C_{T_1, 1}  = 0.001$ for $N'_0 = N'_1 = 100{,}000$, $p^C_{T_0, 1} = p^C_{T_1, 1}  = 0.005$ for $N'_0 = N'_1 = 10{,}000 $, and $p^C_{T_0, 1} = p^C_{T_1, 1}  = 0.03$ for $N'_0 = N'_1 = 1000$.

\paragraph{Misclassification Rates:}
We examine four scenarios for the misclassification rates. In all scenarios, the assay false negative rates are assumed to remain constant across timepoints and are simulated from the following Beta distribution:
\begin{align*}
    &p_{T_0, 0 \mid 1 } = p_{T_1, 0 \mid 1} \sim \text{Beta}(1, 5).
\end{align*}
\noindent In \textbf{Scenario I (no batch effect)}, assay false positive rates remain constant across timepoints and are simulated from the following Beta distribution:
\begin{align*}
    &p_{T_0, 1\mid 0 } = p_{T_1, 1\mid 0 } \sim \text{Beta}(1, 2000).
\end{align*}
\noindent In Scenarios II through IV, the assay false positive rates are distinct across two timepoints.  In \textbf{Scenario II (small batch effect)}, they are simulated as follows:
\begin{align*}
    &p_{T_0, 1\mid 0 } \sim \text{Beta}(1, 2000),~~p_{T_1, 1\mid 0 } \sim \text{Beta}(2, 2000).
\end{align*}

\noindent In \textbf{Scenario III (moderate batch effect)}, they are simulated as follows:
\begin{align*}
    &p_{T_0, 1\mid 0 } \sim \text{Beta}(3, 2000),~~p_{T_1, 1\mid 0 } \sim \text{Beta}(6, 2000).
\end{align*}

\noindent In \textbf{Scenario IV (large batch effect)}, they are simulated as follows:
\begin{align*}
    &p_{T_0, 1\mid 0 } \sim \text{Beta}(1, 2000),~~p_{T_1, 1\mid 0 } \sim \text{Beta}(5, 2000).
\end{align*}

\noindent When the participant is a genuine vaccine responder, the ratio of the vaccine-induced response magnitude to the difference in assay false positive rates approximately ranges from $1$ ($\gamma = 2$ and large batch effect) to $36$ ($\gamma = 10$ and small batch effect).

%\paragraph{Observed Proportions:}
%The observed proportions of positive cells incorporated true positivity and misclassification as:
%\begin{align*}
%    p_{B, \text{obs}} &= p_{B, \text{tp}}(1 - p_{B, \text{fn}}) + (1 - p_{B, \text{tp}})p_{B, \text{fp}}, \\
%    p_{V, \text{obs}} &= p_{V, \text{tp}}(1 - p_{V, \text{fn}}) + (1 - p_{V, \text{tp}})p_{V, \text{fp}}.
%\end{align*}

\paragraph{Observed Positive Counts:}
Positive cell counts are sampled from the following Binomial distributions:
\begin{align*}
    n_0 \sim \text{Binomial}(N_0, p^{\ast}_{T_0, 1}),~~~ 
    n_1 \sim \text{Binomial}(N_1, p^{\ast}_{T_1, 1}), \\
    n'_{0} \sim \text{Binomial}(N'_{0}, p^{C, \ast}_{T_0, 1}),~~~ 
    n'_{1} \sim \text{Binomial}(N'_{1}, p^{C, \ast}_{T_1, 1}),
\end{align*}
where $p^{ \ast}_{T_0, 1} = p_{T_0, 1}\cdot(1 - p_{T_0, 0\mid 1 }) + (1 - p_{T_0, 1})\cdot p_{T_0, 1\mid 0 }$, $p^{\ast}_{T_1, 1} = p_{T_1, 1}\cdot (1 - p_{T_1, 0\mid 1}) + (1 - p_{T_1, 1})\cdot p_{T_1, 1\mid 0}$, $p^{C, \ast}_{T_0, 1} = p^C_{T_0, 1}\cdot (1 - p_{T_0, 0\mid 1 }) + (1 - p^C_{T_0, 1})\cdot p_{T_0, 1\mid 0 }$, and $p^{C, \ast}_{T_1, 1} = p^C_{T_1, 1}\cdot (1 - p_{T_1, 0\mid 1}) + (1 - p^C_{T_1, 1})\cdot p_{T_1, 1\mid 0}$.

For each combination of simulation parameters, we ran the simulation $2,000$ times. Each simulated dataset consists of count data $(n_0, N_0, n_1, N_1, n'_0, N'_0, n'_1, N'_1)$. For each simulated dataset, we calculated three \(p\)-values: (1) the \emph{unadjusted \(p\)-value}, based on a one-sided test using the primary samples and ignoring the paired control data; (2) the \emph{maximally adjusted \(p\)-value}, outlined in Definition \ref{def: BB}; and (3) the \emph{minimally adjusted \(p\)-value}, outlined in Definition \ref{def: minimally adjusted}. We evaluate the type-I error rates and power associated with each method across different data-generating processes. We also report the power of a procedure based on the unobtainable, ground-truth misclassification rates (true power) as a benchmark.

\subsection{Simulations results}
\label{subsec: simu results}
Table~\ref{tab:sim_results} summarizes the average type-I error rates and power for each procedure under various choices of the control sample size $N'_0 = N'_1$ and effect sizes controlled by the parameter $\gamma$. Type-I error rate was computed as the percentage of times a non-responder was declared a vaccine responder, while power was computed as the percentage of times a vaccine responder was declared as such. We observe several consistent trends. First, the maximally adjusted $p$-value maintains the type-I error rates at the nominal level across all data-generating processes, highlighting its validity. Somewhat surprisingly, when the control sample size and the effect size are both reasonably large, the maximally adjusted $p$-value approach has competitive power in all scenarios. For instance, when $N'_0 = N'_1 = 100{,}000$ and $\gamma = 4$, both of which are commonly seen in real ICS datasets for vaccine responders, the maximally adjusted method achieves a power of $0.44$, almost comparable to that of $0.48$ based on the unobtainable true $p$-value, under a small batch effect scenario. It achieves a power of $0.46$ under a large batch effect scenario, again comparable to that of $0.50$ based on the unobtainable true $p$-value. As expected, the maximally adjusted $p$-value shows limited power when the control sample size is small (i.e., when there is limited information available from the control sample) or when the effect size is small.

On the other hand, the minimally adjusted \(p\)-value exhibits favorable performance when there is no batch effect; however, in the presence of a batch effect, its type-I error is inflated as it may underestimate the difference in the misclassification rates between $T_{0}$ and $T_{1}$ (i.e., the true batch effect), leading to falsely declaring a non-respondent a responder. Compared to the unadjusted $p$-value, which has consistently highly inflated type-I error rates across the simulation settings, the minimally adjusted $p$-value starts to exhibit a more favorable type-I error rate when the control sample sizes are large, e.g., when $N'_0 = N'_1 = 100{,}000$, highlighting its ability to partially correct for the type-I error rate inflation when the control samples are large. When the control sample sizes are small, e.g., $N'_0 = N'_1 = 1000$, or when the effect size is small, the maximally adjusted $p$-value may be overly conservative and have low power, while the minimally adjusted $p$-value approach could help detect some useful signals about vaccine responders. 

Overall, we recommend using the maximally adjusted vaccine responder $p$-value for important primary responses when comparing two vaccine regimens. In these scenarios, a valid $p$-value is critical as it serves as a basis for comparing two vaccine regimens. On the other hand, when the control sample size is small, the effect size is expected to be small, or the objective is mostly exploratory (e.g., early-phase development for many HIV vaccines), then using the minimally adjusted $p$-value may be a reasonable alternative. In any case, to provide a more comprehensive and balanced analysis, researchers could report both the maximally and minimally adjusted $p$-values; for instance, one could report one as the primary analysis and the other as a secondary analysis, depending on the nature and requirements of the specific research question.

\begin{table}[H]
\centering
\caption{Type-I error rate and power of each of the three procedures (unadjusted, maximally adjusted, minimally adjusted) under various control sample sizes ($N'_0$) and effect sizes controlled by the parameter $\gamma$. True power is the power based on a test using the unobtainable, ground-truth misclassification rates.  All numbers are multiplied by $100$ (i.e., the unit is a percentage).  }\label{tab:sim_results}
\resizebox{0.95\textwidth}{!}{\begin{tabular}{cccccccccccccc}
 \toprule
  %& & \multicolumn{2}{c}{True} 
  &&\multicolumn{2}{c}{Unadjusted} &&\multicolumn{2}{c}{Maximally adjusted} & &\multicolumn{2}{c}{Minimally Adjusted}  & &True \\
\cline{3-4} \cline{6-7} \cline{9-10} \cline{12-12} 
 $N'_0$ & $\gamma$ & Type-I Err & Power && Type-I Err & Power &  & Type-I Err & Power & &Power\\ 
 \midrule

\multicolumn{12}{c}{\textbf{Scenario I: No batch effect}} \\
   1000 & 2 &3.3 & 41.0& & 0.9 & 2.0 & &3.4 & 40.9 & & 44.4\\ 
  1000 & 4 &2.8 & 46.9& & 1.0 & 6.5 & &3.0 & 46.8  && 48.1\\ 
  1000 & 10 &2.5 & 50.4& & 0.4 & 21.9 & &2.6 & 50.2 && 50.8\\ 
10{,}000 & 2 &2.9 & 42.7& & 0.1 & 12.8 & &2.9 & 42.7 & &46.1\\ 
  10{,}000 & 4 &2.5 & 48.5& & 0.2 & 31.2 & &2.5 & 48.5 &&49.6\\ 
  10{,}000 & 10 &2.1 & 50.7& & 0.2 & 45.3 & &2.3 & 50.7 &&50.8\\ 
  100{,}000 & 2 &2.1 & 40.3& & 0.2 & 32.7 & &2.1 & 40.3 &&44.1\\ 
  100{,}000 & 4 &2.4 & 48.6& & 0.2 & 45.4 & &2.5 & 48.6 && 49.7\\ 
  100{,}000 & 10 & 2.2 & 50.3& & 0.3 & 49.2 & &2.2 & 50.3 &&50.7\\
 \toprule
 \multicolumn{12}{c}{\textbf{Scenario II: Small Batch effect}} \\
 1000 & 2 & 24.3 & 43.9 && 2.1 & 2.2 && 24.2 & 43.9 &&45.0\\ 
  1000 & 4 & 25.8 & 46.1 && 2.4 & 6.8 && 25.8 & 45.9 && 47.3\\ 
  1000 & 10 & 25.0 & 49.2 && 2.0 & 24.0 && 24.9 & 49.2 && 49.7\\ 
  10{,}000 & 2 & 24.4 & 45.1 && 0.5 & 12.4 && 23.9 & 45.2 &&46.0\\ 
  10{,}000 & 4 & 23.6 & 47.8 & &0.2 & 32.5 & &22.9 & 47.9 &&49.5\\ 
  10{,}000 & 10 & 24.2 & 49.6 && 0.4 & 42.3 && 23.2 & 49.6 &&50.2\\ 
100{,}000 & 2 & 24.7 & 44.2 && 0.5 & 33.3 & &16.7 & 44.6 &&45.1\\ 
  100{,}000 & 4 & 25.0 & 46.9 && 0.6 & 44.0 && 17.0 & 47.3 &&48.4\\ 
  100{,}000 & 10 & 24.7 & 49.8 && 0.3 & 48.9 && 16.1 & 50.1 &&50.4\\ 
\toprule
\multicolumn{12}{c}{\textbf{Scenario III: Moderate batch effect}} \\
   1000 & 2 & 37.2 & 43.4 && 0.6 & 0.7 && 37.0 & 43.3 &&42.3\\ 
  1000 & 4 & 34.0 & 51.1 & &0.4 & 4.7 & &33.9 & 51.0 &&51.4\\ 
  1000 & 10 & 35.5 & 50.8 & &0.3 & 21.2 & &35.3 & 50.7 &&50.7\\
   10{,}000 & 2 & 36.9 & 44.2 & &0.3 & 11.0 & &35.3 & 44.0 &&42.8\\ 
  10{,}000 & 4 & 35.6 & 50.5 & &0.6 & 31.2 & &34.5 & 50.6 &&51.1\\ 
  10{,}000 & 10 & 34.0 & 50.8 & &0.3 & 43.5 & &32.7 & 50.8 &&51.0\\ 
  100{,}000 & 2 & 35.7 & 45.7 & &1.0 & 30.4 & &20.1 & 45.0 &&44.2\\ 
  100{,}000 & 4 & 36.2 & 48.3 & &0.8 & 43.7 & &22.9 & 48.5 &&48.6\\ 
  100{,}000 & 10& 34.8 & 51.8 & &1.0 & 49.9 && 20.1 & 51.9 &&51.9\\ 
  \toprule
   \multicolumn{12}{c}{\textbf{Scenario IV: Large batch effect}} \\
1000 & 2 & 45.0 & 49.4 & &1.5 & 2.4 && 44.9 & 49.0 &&45.6\\ 
  1000 & 4 & 44.8 & 49.7 & &1.6 & 6.8 && 44.5 & 49.5 &&48.7\\ 
  1000 & 10 & 43.2 & 50.5 & &2.3 & 21.8 && 42.9 & 50.5 &&50.5\\ 
  10{,}000 & 2 & 43.9 & 49.0 & &0.5 & 13.1 & &42.6 & 48.7 &&44.9\\ 
  10{,}000 & 4 & 44.4 & 50.1 & &0.6 & 34.0 & &43.2 & 50.0 &&49.6\\ 
  10{,}000 & 10 & 43.2 & 51.2 & &0.8 & 43.4 & &41.4 & 51.1 &&50.8\\ 
100{,}000 & 2 & 46.1 & 47.3 & &1.1 & 33.8 & &25.6 & 45.8 &&43.3\\ 
  100{,}000 & 4 & 44.3 & 50.4 & &1.1 & 46.0 & &25.5 & 50.1 &&49.9\\ 
  100{,}000 & 10 & 43.7 & 51.4 & &1.0 & 49.8 & &24.4 & 51.3 &&51.1\\ 
  \bottomrule
    \end{tabular}}
\end{table}

\section{Vaccine responders in the CoVPN 3008 Study}
\label{sec: case study}

\subsection{Description of ICS data and study goals}
\label{subsec: case study paired control}
In the CoVPN 3008 Study, study participants received $1$ or $2$ primary series mRNA-1273 vaccines depending on the baseline SARS-CoV-2 naive status. Specimens from participants in the random immunogenicity subset were collected at baseline (prior to vaccination) and at $4$ weeks post the last vaccination. Flow cytometry was used to examine SARS-CoV-2-specific CD4+ and CD8+ T cell responses using a validated ICS assay that has been previously described \citep{horton2007optimization,de2012omip}. All PBMC samples were assayed against a variety of peptide pools representing the ancestral strain of SARS-CoV-2 and the Omicron BA.4/5 variant. Our analysis will focus on the following: 1) Spike 1 subunit with Omicron BA.4 \& BA.5 variant peptides (BA.4/5 S1), 2) Spike 2 subunit with Omicron BA.4 \& BA.5 variant peptides (BA.4/5 S2), 3) Nucleocapsid with Omicron BA.4 \& BA.5 variant peptides (BA.4/5 N), and 4) Envelope and Membrane with BA.1, BA.2, BA.2.12.1, BA.4 \& BA.5 variant peptides (BA.1/2/4/5 E/M). As a control, cells were incubated with DMSO, the diluent for the peptide pools.

Because the vaccine insert encoded the Spike protein of SARS-CoV-2 but not the Envelope (E), Membrane (M), or Nucleocapsid (N) proteins of SARS-CoV-2 or any of its variant, the vaccine may induce response against BA.4/5 S1 and BA.4/5 S2 but not against BA.4/5 N or BA.1/2/4/5 E/M. If we detect a significant increase in T cell response against BA.4/5 N or BA.1/2/4/5 E/M from baseline ($T_0$) to $4$ weeks post last vaccination ($T_1$), then this could be used as evidence supporting acquisition of an SARS-CoV-2 infection from $T_0$ to $T_1$. Our first goal in this case study is to analyze the T cell response against BA.4/5 N and BA.1/2/4/5 E/M and determine which participants had an incident asymptomatic infection between two specimen collection. Our second goal is to assess vaccine-induced response against BA.4/5 S1 and BA.4/5 S2 among study participants who did not acquire a SARS-CoV-2 infection between $T_0$ and $T_1$.

\subsection{Detecting Omicron infections between two specimen collection timepoints}
\label{subsec: BA.4/5 N}
We first consider $N = 194$ participants with a baseline and post-vaccination measurement of the primary marker, percentage of CD4+ T cells expressing IFN-$\gamma$ and/or IL-2, against BA.4/5 N. One standard preprocessing step in ICS analysis is to filter out measurements that have too few CD4+ T cells. We used a threshold of $10{,}000$ and ended up with $192$ ``per-protocol" participants. The left panel of Figure \ref{fig: case study Omicron infection} plots the unadjusted vaccine responder $p$-values, minimally adjusted vaccine responder $p$-values, and maximally adjusted vaccine responder $p$-values (in the $\text{log}_{10}$-scale) of these $192$ per-protocol participants. Both the maximally and minimally adjusted $p$-values assumed equal false negative rates at $T_0$ and $T_1$ but otherwise did not restrict the assay misclassification rates and let the paired control samples inform them. To aid visualization, we capped all $p$-values less than $1\times 10^{-10}$ at $1\times 10^{-10}$. 

To account for multiple hypothesis testing of $N = 192$ study participants, we used the Benjamini-Hochberg procedure \citep{benjamini1995controlling} to control for the false discovery rate (FDR) at $5\%$. According to the unadjusted method, $18/192 = 9.4\%$ of participants exhibited T cell response to the BA.4/5 N protein and might experience an incident Omicron infection between $T_0$ and $T_1$, while $13/192 = 6.8\%$ and only $5/192 = 2.6\%$ of participants showed T cell response to the BA.4/5 N protein based the minimally adjusted $p$-values and the maximally adjusted $p$-values, respectively. Panel 1 in Table \ref{tab: case study BB BA45 N} summarizes the count data from these $5$ participants (indexed as ID 1-5) who showed T cell response to the BA.4/5 N protein under the maximally adjusted $p$-values. Interestingly, for ID 2 and 3, the maximally adjusted $p$-values showed much stronger evidence of increased T cell response compared to the unadjusted $p$-values, because the paired control samples suggested decreased assay false positivity rate at $T_1$ compared to that at $T_0$. In the downstream analysis, one usual analysis strategy is to define a background-subtracted continuous endpoint at each timepoint and then take the difference in this background-subtracted endpoint. For all these $5$ participants, their differences ($T_1$ minus $T_0$) in the background-subtracted percentage of positive cells are positive, consistent with them being classified as having an Omicron infection. Figure \ref{fig: case study p-val vs magnitude} reinforces this point by plotting both the unadjusted and maximally adjusted p-values against the background-subtracted response magnitude for all $N = 192$ participants. Unlike the unadjusted p-values, the maximally adjusted p-values closely track the background-subtracted response magnitude, with the p-value consistently decreasing as the magnitude of the background-subtracted response increases.

We also examined which participants yielded different conclusions when comparing the minimally adjusted $p$-values with the unadjusted $p$-values. Panel 2 in Table \ref{tab: case study BB BA45 N} summarizes the count data for 7 participants, where the conclusions from the two methods diverged, both controlling for the FDR at 5\%. Interestingly, in nearly every case, there is a substantial difference between the unadjusted and minimally adjusted $p$-values, and the minimally adjusted $p$-values align closely with the background-subtracted response magnitude. For example, for participants ID 1, 2, 3, 4, and 6, the minimally adjusted $p$-values indicate no T cell response to Omicron BA.4/5 N protein; in each case, the background-subtracted response magnitude is either negative (e.g., $-0.05\%$) or very close to zero (e.g., $0.002\%$). However, for ID 7, the minimally adjusted $p$-value suggests an elevated T cell response, and the background-adjusted response magnitude is indeed positive.

We also repeated the analysis for the primary marker against BA.1/2/4/5 E/M; see the right panel of Figure \ref{fig: case study Omicron infection}. The number of participants who showed evidence of increased T cell response against the Omicron Envelope or Membrane proteins dropped from $9$, according to the unadjusted or the minimally adjusted $p$-values, to $3$, according to the maximally adjusted $p$-values, both controlling for the FDR at $5\%$. 

\begin{figure}[ht]
    \centering
   \begin{minipage}{.49\textwidth}
        \centering
\includegraphics[width=0.98\linewidth]{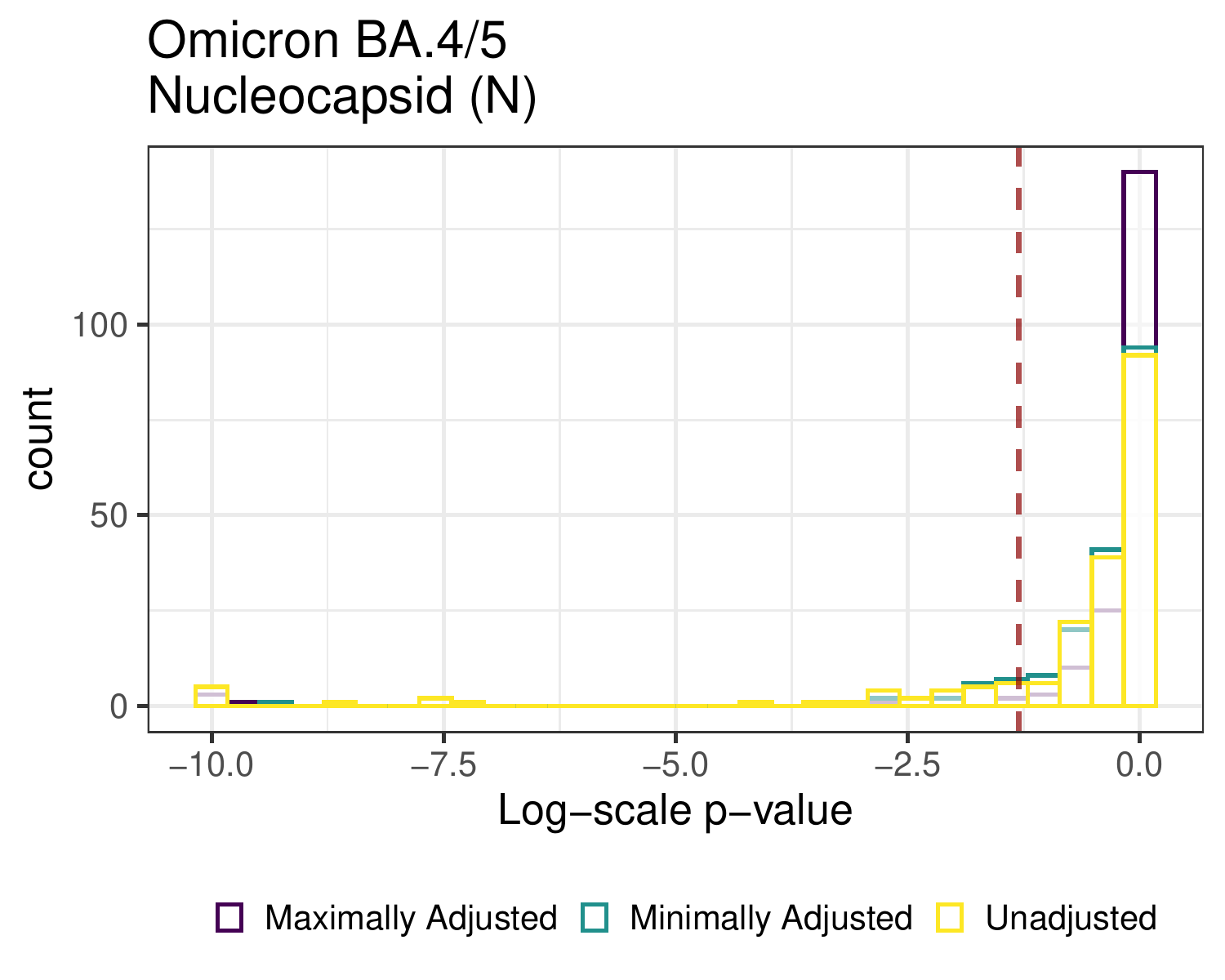}
    \end{minipage}%
    \begin{minipage}{0.49\textwidth}
        \centering
\includegraphics[width=0.98\linewidth]{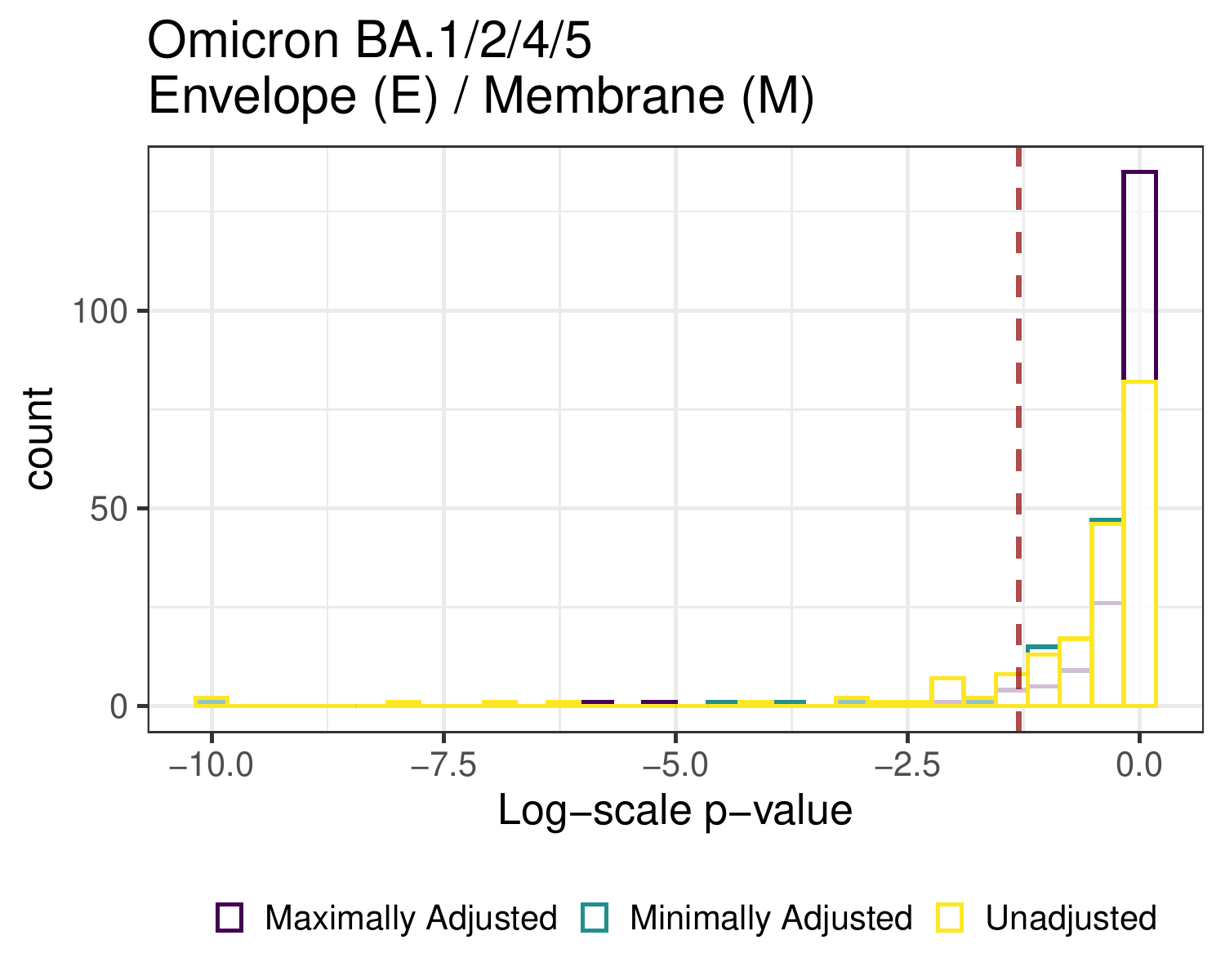}
    \end{minipage}
    \caption{Distribution of the unadjusted (gold), minimally adjusted (turquoise), and maximally adjusted (purple) $p$-values for the percentage of CD4+ T cells expressing IFN-$\gamma$ and/or IL-2 against Omicron BA.4/5 N (left panel) and against Omicron BA.1/2/4/5 E/M (right panel). The red vertical dashed line represents $p = 0.05$.}
    \label{fig: case study Omicron infection}
\end{figure}

\begin{table}[H]
    \centering
    \resizebox{0.9\textwidth}{!}{
    \begin{tabular}{ccccccccc} \hline
     && & && & Maximally  &  &Background\\ 
     \textbf{Panel 1} &\multicolumn{2}{c}{Baseline} & &\multicolumn{2}{c}{Post-vaccination}& Adjusted &  Unadjusted  &Subtracted\\
      & & & & && $p$-value &  $p$-value      &Magnitude\\
     \cline{2-3} \cline{5-6}
      \textbf{ID 1}  & Positive & Negative && Positive & Negative\\
        Primary sample & 4  &51,006 - 4 &&163 &105,179 - 163 & $<1\times 10^{-10}$ & $<1\times 10^{-10}$ & 0.12\%\\
       Control sample & 13 &102,745 - 13 &&84 &213,187 - 84 \\ \hline
       
       \textbf{ID 2}  & Positive & Negative && Positive & Negative\\
        Primary sample &  71 &113,274  &&22 &120,776 & $<1\times 10^{-10}$ & 1 & 0.04\%\\
       Control sample &  148 &116,402 - 148  &&70 &203,093 - 70 \\ \hline
       
       \textbf{ID 3}  & Positive & Negative && Positive & Negative\\
        Primary sample & 54 &133,002 - 54  &&74 &107,231 - 74  & $7\times 10^{-10}$ & 0.001 & 0.06\%\\
       Control sample &  132 &262,821 - 132  &&42 &215,251 - 42 \\ \hline

       \textbf{ID 4}  & Positive & Negative && Positive & Negative\\
        Primary sample &  17  &27,563 - 17  &&80  &27,968 & $<1\times 10^{-10}$ & $1\times 10^{-10}$ & 0.25\%\\
       Control sample &  34  &79,041 - 34   &&8  &55,784 - 8 \\ \hline

       \textbf{ID 5}  & Positive & Negative && Positive & Negative\\
        Primary sample &  12 &118,019 - 12 &&102 &135,183 - 102 & $2\times 10^{-10}$ & $<1\times 10^{-10}$ & 0.05\%\\
       Control sample &  55 &241,697 - 55  &&92 &275,685 - 92 \\ \hline

 && & && & Minimally  &  &Background\\ 
     \textbf{Panel 2} &\multicolumn{2}{c}{Baseline} & &\multicolumn{2}{c}{Post-vaccination}& Adjusted &  Unadjusted  &Subtracted\\
      & & & & && $p$-value &  $p$-value      &Magnitude\\
     \cline{2-3} \cline{5-6}

      \textbf{ID 1}  & Positive & Negative && Positive & Negative\\
        Primary sample & 109 & 141,174 - 109 && 515 & 160,080 - 515 & 1 &  $<1\times 10^{-10}$ & -0.05\%\\
       Control sample & 191 & 266,842 - 191 && 1,175 & 321,473 - 1,175 \\ \hline
       \textbf{ID 2}  & Positive & Negative && Positive & Negative\\
        Primary sample &  52 & 73,398 - 52 && 119 & 70,127 - 119 &0.03 & $3\times 10^{-8}$ & $0.003\%$\\
       Control sample &  10 & 75,920 - 10 && 149 & 136,953 - 149 \\ \hline
       \textbf{ID 3}  & Positive & Negative && Positive & Negative\\
        Primary sample & 111 & 80,524 - 111 && 212 & 81,626 - 212  &0.15 & $2\times 10^{-8}$ & $0.002\%$\\
       Control sample &  35 & 183,124 - 35 && 239 & 171,449 - 239 \\ \hline

       \textbf{ID 4}  & Positive & Negative && Positive & Negative\\
        Primary sample &  69 & 142,560 - 69 && 116 & 102,021 - 116 & 0.08 & $3\times 10^{-9}$ & $0.003\%$\\
       Control sample &  29 & 314,409 - 29 && 141 & 196,983 - 141 \\ \hline

       \textbf{ID 5}  & Positive & Negative && Positive & Negative\\
        Primary sample &  21 & 117,374 - 21 &&  31 & 84,036 - 31 & 0.004 & 0.004 & 0.02\%\\
       Control sample &  22 & 221,106 - 22 &&  13 & 171,825 - 13 \\ \hline

        \textbf{ID 6}  & Positive & Negative && Positive & Negative\\
        Primary sample &  50 & 103,686 - 50 &&  96 & 120,586 - 96 &1 & 0.002 & -0.08\%\\
       Control sample& 133 & 208,117 - 133 && 421 & 238,717 - 421 \\ \hline

       \textbf{ID 7}  & Positive & Negative && Positive & Negative\\
        Primary sample & 71 &113,274 - 71 &&22 &120,776 - 22 & $<1\times 10^{-10}$ & 1 & 0.05\%\\
       Control sample & 148 & 116,402 - 148  &&70  &203,093 - 70 \\ \hline

    \end{tabular}}
    \caption{\textbf{Panel 1}: Count data from $5$ participants who were considered exhibiting an increased level of T cell response to the Omicron BA.4/5 N protein from $T_0$ to $T_1$, according to the maximally adjusted $p$-values controlling the FDR at $5\%$. \textbf{Panel 2}: Count data from $7$ participants. ID 1 to ID 6 were considered to exhibit an increased level of T cell response to the Omicron BA.4/5 N protein from $T_0$ to $T_1$, according to the unadjusted but not the minimally adjusted $p$-values, controlling the FDR at $5\%$. ID 7 was considered to exhibit an increased level of T cell response to the Omicron BA.4/5 N protein from $T_0$ to $T_1$, according to the minimally adjusted but not the unadjusted $p$-values.}
    \label{tab: case study BB BA45 N}
\end{table}

\begin{figure}[ht]
    \centering
   \begin{minipage}{.49\textwidth}
        \centering
\includegraphics[width=0.98\linewidth]{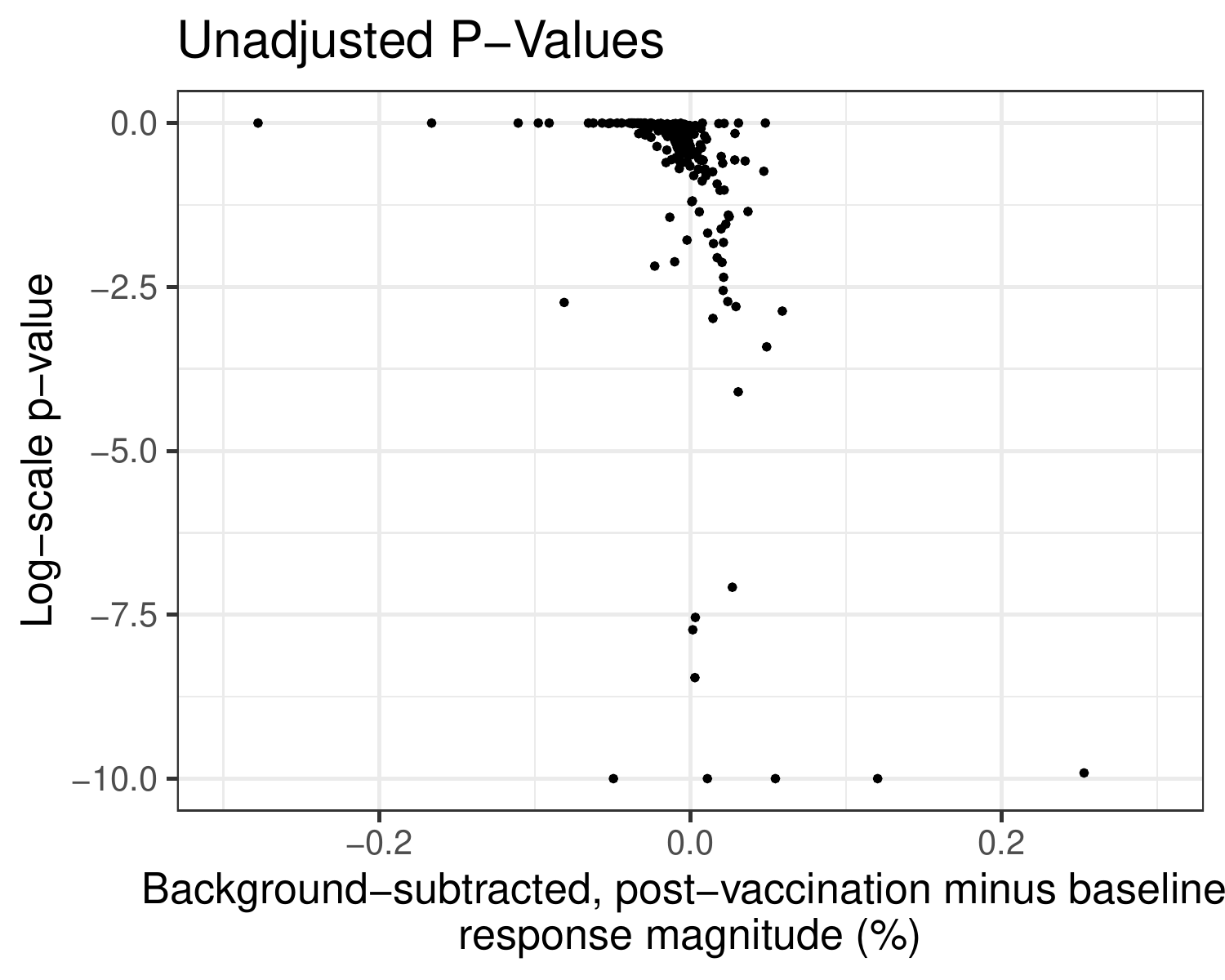}
    \end{minipage}%
    \begin{minipage}{0.49\textwidth}
        \centering
\includegraphics[width=0.98\linewidth]{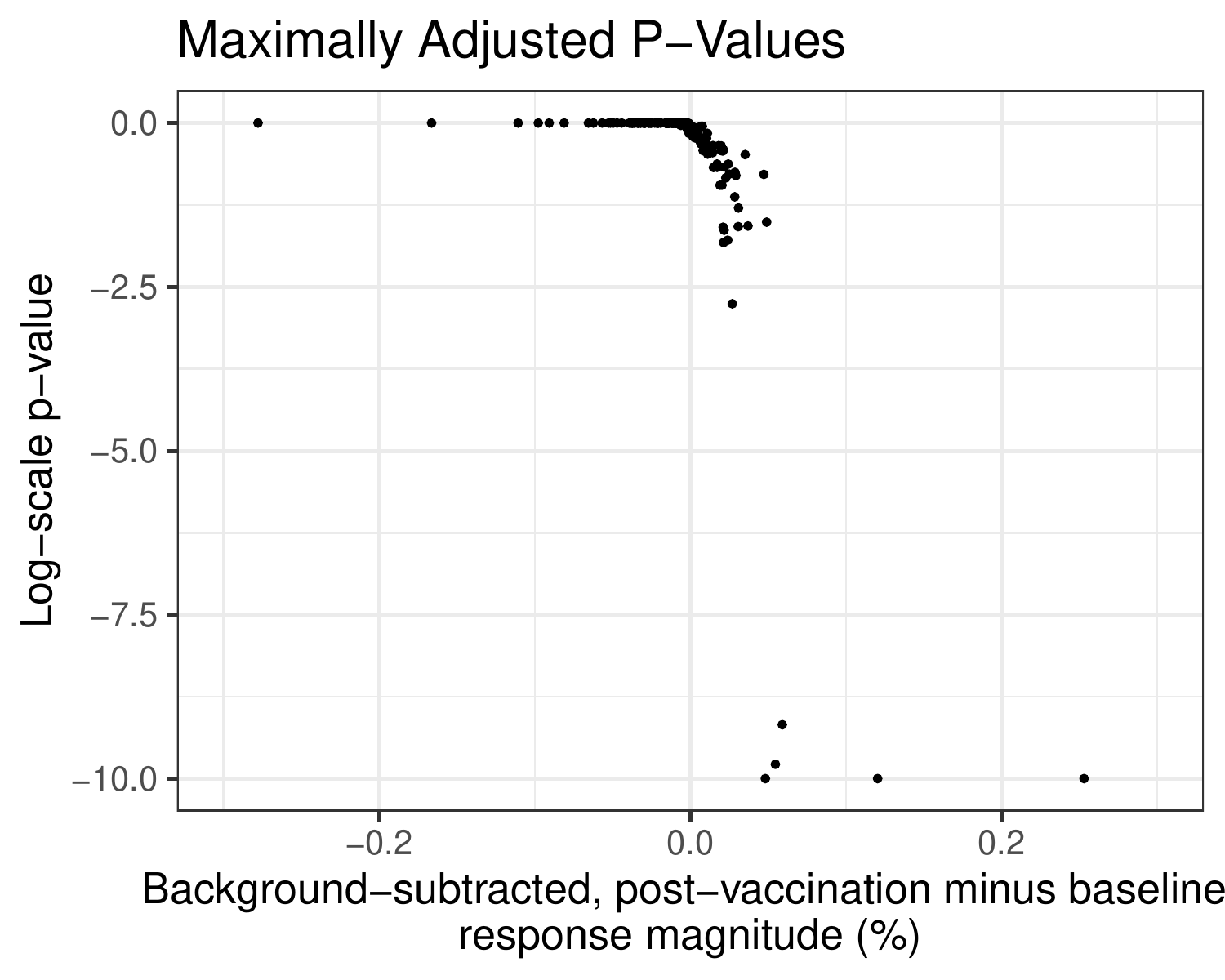}
    \end{minipage}
    \caption{Log-scale unadjusted p-values (left panel) and maximally adjusted p-values (right panel) for the percentage of CD4+ T cells expressing IFN-$\gamma$ and/or IL-2 against Omicron BA.4/5 N versus the background-subtracted response magnitude.}
    \label{fig: case study p-val vs magnitude}
\end{figure}

%As we examine the raw count data carefully, the batch effect could sometimes be dramatic (e.g., ID 1 and 3). 

%\textcolor{red}{For all these $7$ participants, the minimally adjusted $p$-values appear to reflect faithfully the batch effect information contained in the paired control samples. }

\subsection{Vaccine induced response against Omicron BA.4/5 Spike protein}
\label{subsec: case study BA.4/5 S}
We analyzed the proportion of CD4+ T cells expressing IFN-$\gamma$ and/or IL-2 in response to the Omicron BA.4/5 Spike protein subunits S1 and S2. We excluded 17 participants who showed evidence of Omicron infection (T cell response against Omicron BA.4/5 N and/or BA.1/2/4/5 E/M) between $T_0$ and $T_1$ based on the minimally adjusted $p$-values in Section \ref{subsec: BA.4/5 N}. We chose to use the minimally adjusted $p$-value, rather than the maximally adjusted $p$-value, as a less conservative criterion for defining an incident Omicron infection. This approach was intended to be more conservative in defining a cohort free from evidence of Omicron infection, ensuring that any observed response to the Omicron BA.4/5 Spike protein is likely to be vaccine-induced. 

We ended up with $175$ per-protocol participants. Using the maximally adjusted $p$-values and controlling for the FDR at $5\%$, $155/175 = 88.6\%$ and $155/175 = 88.6\%$ of participants were determined to be a vaccine responder who exhibited vaccine-induced immunogenicity against BA.4/5 S1 and BA.4/5 S2, respectively. When we controlled the FDR at $1\%$, then $150/175 = 85.7\%$ and $145/175 = 82.9\%$ of participants were determined a vaccine responder to subunit S1 and S2, respectively. Figure \ref{fig: case study p-val vs magnitude S1} plots the unadjusted and maximally adjusted $p$-values (Omicron BA.4/5 Spike subunit S1) versus the background-subtracted response magnitude for $N = 175$ participants. Again, the maximally adjusted $p$-values appear to better track the response magnitude compared to the unadjusted analysis.

\begin{figure}[ht]
    \centering
   \begin{minipage}{.49\textwidth}
        \centering
\includegraphics[width=0.98\linewidth]{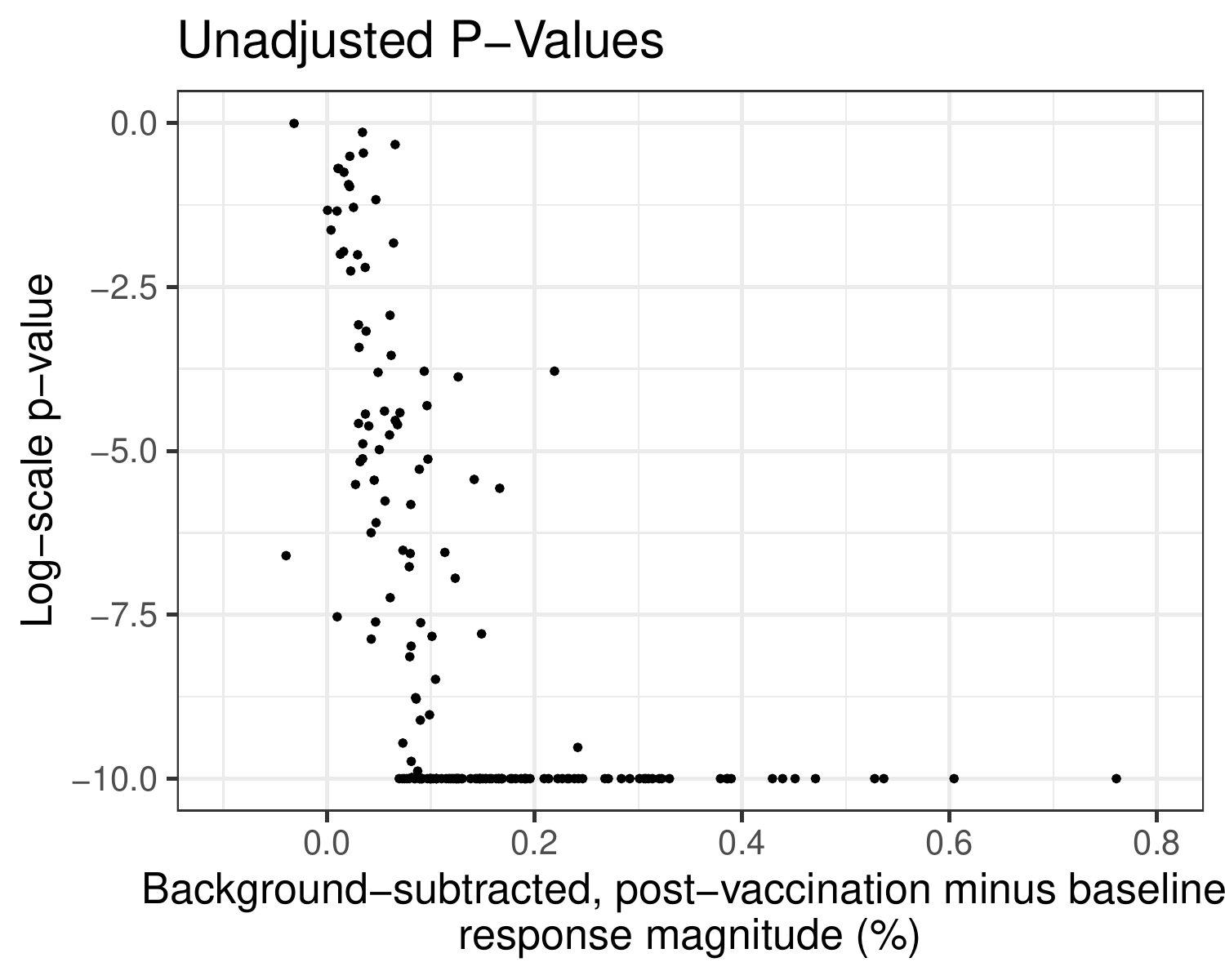}
    \end{minipage}%
    \begin{minipage}{0.49\textwidth}
        \centering
\includegraphics[width=0.98\linewidth]{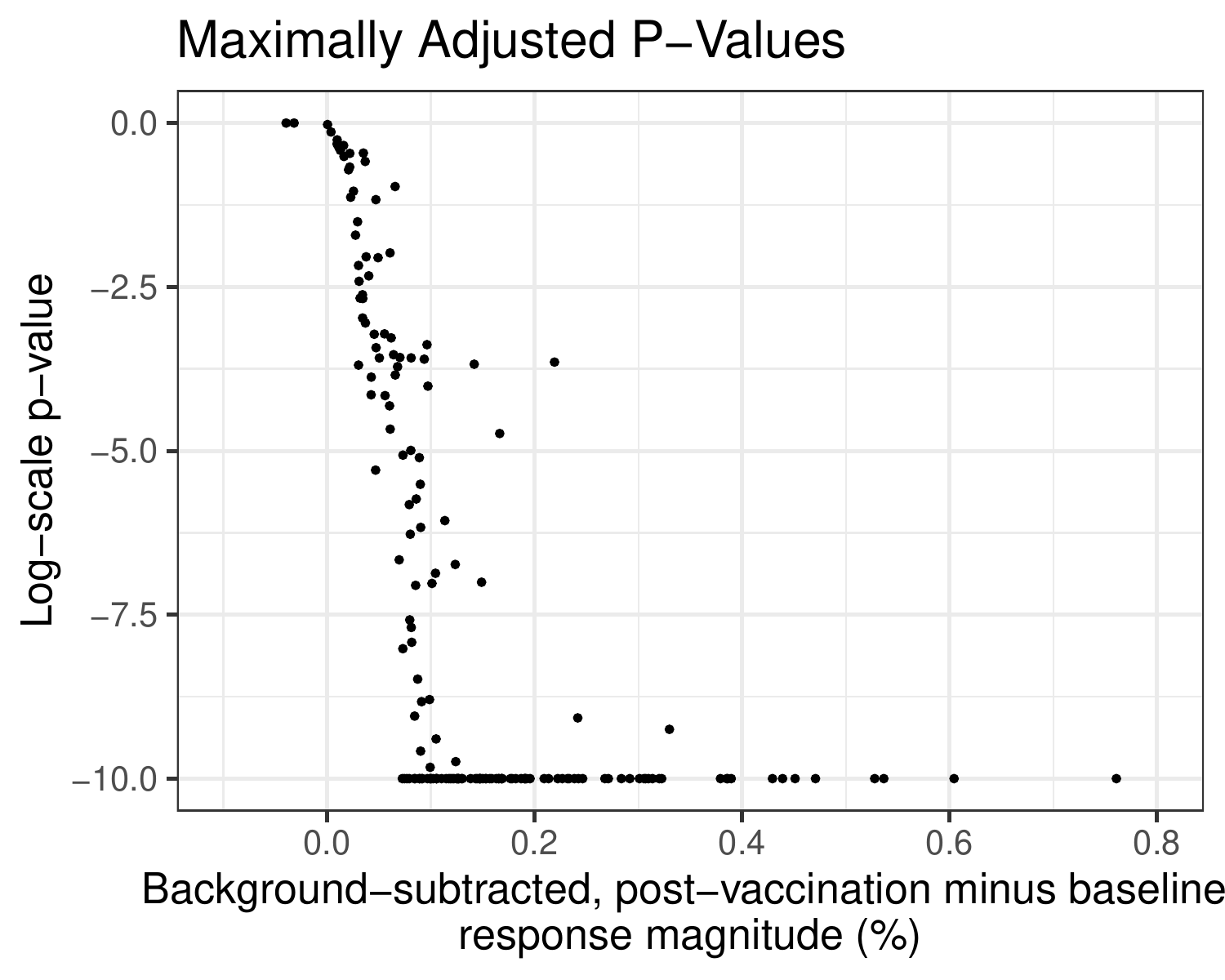}
    \end{minipage}
    \caption{Log-scale unadjusted p-values (left panel) and maximally adjusted p-values (right panel) for the percentage of CD4+ T cells expressing IFN-$\gamma$ and/or IL-2 against Omicron BA.4/5 Spike subunit S1 versus the background-subtracted response magnitude.}
    \label{fig: case study p-val vs magnitude S1}
\end{figure}

\section{Discussion}
\label{sec: discussion}
In this work, we proposed a novel statistical framework to incorporate the paired control readouts into determining a vaccine responder using single cell-level assay data measured at baseline and post-vaccination. The key assumption in our methods is that the primary samples and their paired control samples measured in the same assay run share the same assay misclassification rates. We leverage the paired control samples to infer the assay misclassification rates and then incorporate the inferred assay misclassification rates into the vaccine responder $p$-values. We studied two strategies: the maximally adjusted $p$-value reports the worst-case vaccine responder $p$-value across all plausible assay misclassification rates, as informed by the control samples, while the minimally adjusted $p$-value provides a $p$-value that is not contradicted by the paired control samples and aligns most closely with the primary analysis $p$-value. We explored the trade-off between type-I and type-II errors of these two $p$-value adjustment methods. Based on the simulation studies and the case study, the maximally adjusted $p$-value, in addition to being always valid, exhibits statistical power comparable to the unobtainable oracle $p$-value (assuming known assay misclassification rates) in many relevant data-generating processes. The minimally adjusted $p$-value has improved validity compared to the unadjusted $p$-value and could be a reasonable alternative when the maximally adjusted $p$-value lacks power. 

Overall, for a stringent comparison of immunogenicity of primary markers, such as the immune response against vaccine-matched antigens, we recommend reporting the maximally adjusted $p$-values. If researchers are interested in exploring additional discoveries, such as asymptomatic Omicron infection as discussed in our case study (Section \ref{subsec: BA.4/5 N}), then findings from the minimally adjusted $p$-values could be of interest and serve as candidate hypotheses for further investigation.

Our proposed framework can be extended in several ways. First, in the current work, we model the number of cells as being drawn from a binomial distribution, which does not account for potential correlations among cells within the same assay run. These correlations could lead to a smaller effective sample size. One approach to better address this issue is to incorporate technical replicates --- primary or negative control samples taken within the same assay run. Since these replicates are processed under the same conditions, we would expect to detect no statistical difference between them. Any observed, statistically significant differences could then be attributed to the correlations among cells, providing insights into the resulting sample size inflation. Second, rather than assuming that assay misclassification rates are uniform within a single assay run, one could model the relationship between these rates and the number of cells in each run. Third, while all participants in the setting discussed in this article received vaccination, many other studies involve randomizing participants to either an active vaccination or a placebo group (\citealp{hudgens2006causal}, \citealp{gabriel2020unified}, \citealp{gilbert2023controlled}, \citealp{chen2024role}). The randomized design, along with the inclusion of a placebo arm, could offer valuable additional information to further strengthen the accuracy of vaccine responder classifications. Finally, beyond the primary vaccine-responder analysis, the paired control sample could be utilized in sensitivity analyses for assay misclassification. This could be done by integrating the methodology proposed in our work with established sensitivity analysis frameworks for outcome misclassification (e.g., \citealp{shepherd2008does}, \citealp{imai2010causal}, \citealp{gilbert2016misclassification}, \citealp{heng2025sensitivity}).

\section*{Acknowledgment}
The authors thank the participants and study staff of the CoVPN 3008 trial study for their dedication and contributions. 
The authors thank Dr. Peter Gilbert, Dr. Andrew Fiore-Gartland, Dr. Erica Andersen-Nissen and Dr. Sharon Khuzwayo for their helpful comments. 

\section*{Supplementary Materials}

The online supplemental materials include detailed proofs of all theoretical results in this work and additional tables and discussions. The \texttt{R} code for implementing the methods and replicating the simulation studies is available on the GitHub page (\url{https://github.com/Zhe-Chen-1999/Debias-Positivity-Call}).

\bibliographystyle{apalike}
\bibliography{paper-ref}

\end{document}